\documentclass[12pt]{amsart}
\usepackage[margin=1in]{geometry}
\usepackage{amsfonts}
\usepackage{amssymb}
\usepackage{mathrsfs}
\usepackage[dvips]{graphics}
\usepackage{epsfig}
\usepackage{euscript}
\usepackage{color}
\usepackage{bbm}
\usepackage[colorlinks ,pdfstartview=FitB]{hyperref}
\hypersetup{allcolors=blue}

\newtheorem{thm}{Theorem}[section]

\newtheorem{lem}[thm]{Lemma}

\newtheorem{proposition}[thm]{Proposition}
\theoremstyle{definition}

\theoremstyle{remark}

\numberwithin{thm}{section}
\DeclareMathOperator{\RE}{Re}
\DeclareMathOperator{\IM}{Im}

\newcommand{\R}{{\mathord{\mathbb R}}}
\newcommand{\N}{{\mathord{\mathbb N}}}

\newcommand{\I}{\mathcal{I}}

\newcommand{\W}{\mathcal{W}}
\newcommand{\X}{\mathcal{X}}

\newcommand{\C}{{\mathord{\mathbb C}}}
\newcommand{\Z}{{\mathord{\mathbb Z}}}
\newcommand{\E}{{\mathord{\mathbb E}}}

\def\idty{{\mathchoice {\mathrm{1\mskip-4mu l}} {\mathrm{1\mskip-4mu l}} %
{\mathrm{1\mskip-4.5mu l}} {\mathrm{1\mskip-5mu l}}}}

\DeclareMathOperator{\tr}{tr}

\DeclareMathOperator{\supp}{supp}
\DeclareMathOperator{\diag}{diag}

\DeclareMathOperator{\Ran}{Ran}
\begin{document}

\title[On Localized Excitations]{On the Regime of Localized Excitations \\
for Disordered Oscillator Systems}

\author[H. Abdul-Rahman]{Houssam Abdul-Rahman}
\address{Department of Mathematics\\
University of Arizona\\
Tucson, AZ 85721, USA}
\email{houssam@math.arizona.edu}
\author[R. Sims]{Robert Sims}
\address{Department of Mathematics\\
University of Arizona\\
Tucson, AZ 85721, USA}
\email{rsims@math.arizona.edu}
\author[G. Stolz]{G\"unter Stolz}
\address{Department of Mathematics\\
University of Alabama at Birmingham\\
Birmingham, AL 35294 USA}
\email{stolz@uab.edu}

\date{\today}

\begin{abstract}

We study quantum oscillator lattice systems with disorder, in arbitrary dimension, requiring only partial localization of the associated effective one-particle Hamiltonian. This leads to a many-body localized regime of excited states with arbitrarily large energy density. We prove zero-velocity Lieb-Robinson bounds for the dynamics of Weyl operators as well as for position and momentum operators restricted to this regime. Dynamical localization is also shown in the form of quasi-locality of the time evolution of local Weyl operators and through exponential clustering of the dynamic correlations of states with localized excitations. 

\end{abstract}

\maketitle

%
%

\allowdisplaybreaks

\section{Introduction}

Systems of interacting quantum harmonic oscillators are one of the simplest toy-models in which many-body localization (MBL) effects due to disorder can be rigorously studied. They may be considered as a counterpart to the XY spin chain in that the former can be mapped to a free lattice boson system while the latter maps to a free lattice fermion system. One interesting feature of oscillator systems is that this mapping works in arbitrary dimension (while the required Jordan-Wigner transform for the XY chain is limited to dimension one). Another conceptional difference is that the local Hilbert space in oscillator systems is infinite-dimensional, requiring technical adjustments to the finite-dimension spin systems case (such as working on the Weyl algebra of local observables). In both models, however,  many-body localization properties can be studied through an effective one-particle Hamiltonian, with one of the tasks consisting in properly relating one-particle and many-particle concepts.

We refer to the survey \cite{ANSS} for results on the disordered XY chain. Another relatively simple model where the emergence of certain MBL-type properties can be studied via reduction to an effective one-body Hamiltonian is the Tonks-Girardeau gas \cite{SW}. We also mention the disordered Holstein model, describing an Anderson model coupled to a system of free bosons via a tracer particle, where localization properties have bee proven in \cite{MaviSchenker}. This model does not reduce to a one-particle operator and requires a more thorough and original fractional moment analysis of correlations between the series of bands arising from the bosonic modes.

Oscillator lattice systems are the standard model for phonons, the vibrational degrees of freedom in crystal lattices. The importance of disorder effects on such systems has long been realized, see the classical works \cite{Rieder1967,MatsudaIshii,CasherLebowitz} for first studies of the non-interacting case. Interacting oscillator systems have been studied more recently, initially under the assumption of a stable ground state gap, see \cite{CE,SCW,CSE,NRSS,Nourr} for bounds on transport in the form of Lieb-Robinson estimates as well as results on exponential clustering of ground state correlations. There were also multiple studies of entanglement bounds in gapped oscillator systems, a topic which we will not discuss here, so we point to \cite{NSS2} and, in particular, its bibliography for more references.

One of the key properties of the disordered oscillator systems studied here is that they do not have a stable gap. The central question could be phrased as asking if these systems still have a `mobility gap' which leads to localization properties.

We will study a $\nu$-dimensional disordered oscillator system for which the effective Hamiltonian is the $\nu$-dimensional Anderson model, with arbitrary $\nu$. Much of the prior work on this model (e.g.\ \cite{NSS1}, \cite{NSS2},  \cite{ARSS17} and \cite{AR18}, which we will compare with our new results in more detail below) has focused on the case of large disorder where the effective Anderson model is fully localized. 
However, in dimension $\nu\ge 2$ and for small disorder, the Anderson model is only known to be localized near the spectral edges and, at least for $\nu \ge 3$, expected to have an extended states regime near the center of the spectrum.

It is thus natural to ask if and how such a partially localized regime of the effective Hamiltonian leads to a corresponding many-body localization regime for the oscillator system. Describing this regime, which we will refer to as the {\it regime of localized excitations}, is our main goal here. For this we have to study MBL properties of the corresponding excited states of the oscillator system, going beyond earlier works which generally considered the ground state or thermal states (e.g., in the form of rapid decay of correlations or small entanglement). Our results for excited states are new even for fully localized systems.

All our results will also account for {\it dynamical} many-body localization properties. Generally, this is done by studying the evolution of local observables under the Heisenberg dynamics of the system. That the support of these observables remains essentially bounded for all times, up to exponentially small tails, can be expressed in the form of zero-velocity Lieb-Robinson bounds or, more directly, by the possibility to closely approximate the time-evolved observable with a strictly local observable (referred to in the following as {\it quasi-locality} of the time evolution).

That our proof of localization for this many-body system is valid only in the regime of localized excitations (in particular, it does not exclude many-body delocalized states in other parts of the Hilbert space) will be reflected by considering the Heisenberg dynamics of restricted observables, i.e., two-sided projected onto the subspace of localized excitations. This is similar to recent work on dynamical localization of the droplet spectrum in the disordered XXZ spin chain \cite{EKS, EKS2}, where the Heisenberg dynamics of observables restricted to the droplet regime was studied (see also \cite{BW} for closely related results).

In Section~\ref{sec:model} we will introduce the model and the assumptions we make on the disorder, reviewing, in particular, the localization properties of the Anderson model which we will need. We then, in Section~\ref{sec:freebosons}, recall the reduction of oscillator systems to a free boson system and, in this context, introduce the regime of localized excitations, corresponding to localized states of the effective Hamiltonian.

Section~\ref{sec:prelim} provides some important tools: In Section~\ref{sec:Weyl} we discuss the Weyl operators and their restriction to the regime of localized excitations (some more background on Weyl operators is collected in Appendix~\ref{sec:Weylops}). In Section~\ref{sec:effdyn} we show how one-particle localization of the effective Hamiltonian can be restated as localization of the effective dynamics of the Weyl operators. This provides the crucial link between one-body dynamics and many-body dynamics in our work.

In Sections~\ref{sec:LRBs}, \ref{sec:quasiloc} and \ref{sec:cordecay} we state and prove our three main results. We start with a zero-velocity Lieb-Robinson bound for the Heisenberg evolution of observables restricted to the regime of localized excitations, both for the Weyl operators and for local position and momentum operators, see Section~\ref{sec:LRBs}. This is followed in Section~\ref{sec:quasiloc} by establishing quasi-locality of the evolution of the Weyl operators, again using restriction to localized excitations. Note that, due to the restrictions on the observables under consideration, there is no obvious equivalence between 
Lieb-Robinson bounds and the corresponding quasi-locality estimates. Finally, we prove exponential decay of dynamic correlations for the Weyl operators in eigenstates with localized excitations. This is the content of Section~\ref{sec:cordecay}. 

We include two additional appendices. In Appendix~\ref{sec:nondeg} we show that the disordered oscillator systems considered here almost surely have non-degenerate spectrum. We use this in Sections~\ref{sec:quasiloc} and \ref{sec:cordecay}, but also consider this fact and its proof as being of some independent interest. In Appendix~\ref{App:energydensity} we show that the regime of localized excitations (for which MBL is established here) is extensive in energy, i.e., it allows for states with arbitrarily large positive energy density. In fact, we give an explicit formula for the maximal energy density (at any fixed number of excitations per site) in terms of the density of states of the infinite volume Anderson model. 

At this point it remains an open problem to prove an area law for the bipartite entanglement entropy of states in the regime of localized excitations. For the ground state and thermal states this was accomplished in \cite{NSS2}, assuming a fully localized system (e.g.\ large disorder for the Anderson model). A difficulty in extending this result to excitations is that excited states in oscillator systems are no longer quasi-free (as is the case for the XY chain, a fact used in \cite{ARS15} to prove an area law for the disordered XY chain, uniformly for all eigenstates). However, \cite{AR18} has identified a class of mixed non-quasi-free states in disordered oscillator systems where an area law in terms of the disorder-averaged logarithmic negativity can be shown. In forthcoming work we will address related results concerning area laws for the dynamic entanglement of a large class of states under quantum quenches (comparable to a result in \cite{ARNSS} for the XY chain).

\vspace{.5cm}

\noindent {\bf Acknowledgments:} G.~S.\ gratefully acknowledges hospitality and support at the Centre de Recherches Math\'ematiques of the Universit\'e de Montr\'eal, where part of this work was done during the Thematic Semester on Mathematical challenges in many-body physics and quantum information.

\section{Disordered quantum oscillator systems} \label{sec:oscsystems}

\subsection{Model and assumptions} \label{sec:model}

For any dimension $\nu\geq 1$, we consider harmonic oscillator systems over arbitrary finite rectangular boxes $\Lambda:= [a_1,b_1]\times\ldots\times[a_\nu,b_\nu] \subset \Z^\nu$, with $\ell^1$-distance $|\cdot|$, given by the Hamiltonian
\begin{eqnarray} \label{Hamiltonian}
H_\Lambda & = & \sum_{x\in\Lambda}(p_x^2+k_x q_x^2) + \sum_{x,y\in \Lambda, |x-y|=1} (q_x-q_y)^2\\
& = & \sum_{x\in\Lambda} p_x^2 + q^T h_{\Lambda} q \notag
\end{eqnarray}
in the Hilbert space
\begin{equation} \label{Hilbertspace}
\mathcal{H}_\Lambda=\bigotimes_{x\in\Lambda}\mathcal{L}^2(\mathbb{R})=\mathcal{L}^2(\mathbb{R}^\Lambda).
\end{equation}
Here  $q_x$ and $p_x=-i\partial/\partial q_x$ are the position and momentum operators at the sites $x\in \Lambda$, and $q=(q_1,q_2,\ldots,q_\Lambda)^T$ is viewed as a column vector (and thus $q^T=(q_1,\ldots,q_\Lambda)$ a row vector). The effective Hamiltonian of the system is
\begin{equation} \label{effHam}
h_\Lambda = h_{0,\Lambda} + k,
\end{equation}
where $h_{0,\Lambda}$ is the (non-negative semi-definite) graph Laplacian
\begin{equation} \label{graphLap}
(h_{0,\Lambda} f)(x) := \sum_{y\in \Lambda:|x-y|=1} (f(x)-f(y))
\end{equation}
on $\ell^2(\Lambda)$ and $k$ the multiplication operator by a positive potential $k:\Lambda \to (0,\infty)$. This means that $h_\Lambda$ is strictly positive definite, thus assuring positivity and discreteness of the spectrum of $H_\Lambda$ by standard results, e.g.\ \cite{RS2}. Also, $\|h_{\Lambda}\| \le 4\nu + \|k\|_{\infty}$.

For our results below we assume that
\begin{equation} \label{iid}
(k_x)_{x\in\Lambda} \: \mbox{are i.i.d.\ random variables},
\end{equation}
so that $h_\Lambda$ is the $\nu$-dimensional finite volume Anderson model. We further assume that the random variables
\begin{equation} \label{density}
\begin{array}{l} k_x, \:x\in \Lambda,\: \mbox{have a.c.\ distribution $d\mu(k_x)=\rho(k_x)dk_x$}, \\ \mbox{with bounded density $\rho$, supported on} \,\:[0,k_{max}]. \end{array}
\end{equation}

The model and assumptions \eqref{Hamiltonian}, \eqref{effHam}, \eqref{iid} and \eqref{density} will remain fixed throughout this work and all our results below refer to this situation. By $\E(\cdot)$ we will denote averaging over these random variables, i.e., with respect to the product measure $d\mathbb{P} = \prod_{x\in \Lambda} d\mu$ on $\R^{\Lambda}$.

By \eqref{effHam} and \eqref{density} we have the almost sure norm bound $\|h_\Lambda\| \le  4\nu + k_{max}$, uniformly in $\Lambda$. Note that, while $h_\Lambda$ is almost surely strictly positive definite, allowing the support of the random variables $k_x$ to contain $0$ means that $\inf \sigma(h_\Lambda)$ does not have a deterministic and $|\Lambda|$-independent positive lower bound, so that $h_\Lambda^{-1}$ almost surely exists but does not have a deterministic upper norm bound. A consequence of this is that the many-body Hamiltonian $H_\Lambda$ does not have a stable ground state gap, compare \eqref{eigenvalues} below.

Absolute continuity of the $k_x$ implies that
\begin{equation} \label{simplespec}
h_\Lambda \: \mbox{and} \:\: H_\Lambda \:\: \mbox{almost surely have simple spectrum},
\end{equation}
i.e., all their eigenvalues are non-degenerate. While this is well known for the one-body operator $h_\Lambda$, we provide a proof of the non-degeneracy of the many-body operator $H_\Lambda$ in Appendix~\ref{sec:nondeg}. We will use these properties below.  In the case of the many-body Hamiltonian we also consider this fact and its proof as being of some independent interest.

All our results below will be based on the following form of localization for the Anderson model $h_\Lambda$: There are constants $\lambda_0>0$, $C<\infty$ and $\mu>0$, independent of $\Lambda$, such that
\begin{equation}\label{def:eig-corr}
\mathbb{E}\left(\sup_{|u|\leq 1} \left|\langle\delta_x,h_{\Lambda}^{-1/2} u(h_\Lambda)\chi_{[0,\lambda_0]}(h_{\Lambda})\delta_y\rangle \right|\right)\leq C e^{-\mu|x-y|}
\end{equation}
for all $x, y \in \Lambda$. Here $\chi_{[0,\lambda_0]}(h_\Lambda)$ refers to the spectral projection for $h_\Lambda$ onto $[0,\lambda_0]$ and the supremum is over all Borel functions with pointwise bound $|u|\le 1$, with $u(h_\Lambda)$ and $h_{\Lambda}^{-1/2}$ defined via the functional calculus. Without the `singular' factor $h_{\Lambda}^{-1/2}$ the bound \eqref{def:eig-corr} is generally referred to as localization of eigenfunction correlators and well known to hold under our assumptions, in any dimension $\nu$ and on a non-trivial interval $[0,\lambda_0]$. The stronger localization bound \eqref{def:eig-corr} was shown in Appendix~A of \cite{NSS1}, also in arbitrary dimension and under assumption \eqref{density}. 

We mention that all of our results could be extended to more general disordered oscillator systems than \eqref{Hamiltonian}, e.g.\ with random masses weighing the kinetic energies $p_x^2$ or random couplings at the interactions $(q_x-q_y)^2$, as long as localization of the effective Hamiltonian in the form \eqref{def:eig-corr} can be verified. Limiting our discussion to the case of random spring constants $k_x$ is mostly due to the fact that this can most easily be referenced for the Anderson model (with disordered potential).

\subsection{Free boson systems and the regime of localized excitations} \label{sec:freebosons}

Here we recall the basic concepts behind the reduction of interacting harmonic oscillators to free boson systems. In particular, this identifies $h_\Lambda$ as the effective one-particle Hamiltonian governing the many-body system $H_\Lambda$.
In this language we will then introduce the subspace of $\mathcal{H}_\Lambda$ which will turn out to describe the many-body localized phase.

Denote by $0<\gamma_1 \le \gamma_2 \le ... \le \gamma_{|\Lambda|}$ the positive square roots of the eigenvalues $\gamma_j^2$ of $h_\Lambda=h_{0,\Lambda}+k$, in non-decreasing order and counted with multiplicity. By elementary perturbation theory of real hermitean matrices, e.g.\ \cite{Kato}, the $\gamma_j$ are continuous in $k\in (0,\infty)^{\Lambda}$. One may also choose a corresponding orthonormal basis of real eigenvectors $\varphi_j$ of $h_\Lambda$, $1\le j \le |\Lambda|$, which depends continuously on $k$. Thus the orthogonal mapping $\mathcal{O}: \R^{|\Lambda|} \to \R^{\Lambda}$ given by $(\mathcal{O} v)(x) = \sum_j \varphi_j(x) v_j$, $x\in \Lambda$, diagonalizes $h_\Lambda$,
\begin{equation} \label{diagh}
{\mathcal O}^T h_\Lambda {\mathcal O} = \gamma^2.
\end{equation}
where $\gamma = \diag(\gamma_j)$ and $({\mathcal O}^T f)(j)= \sum_{x\in \Lambda} f(x) \varphi_j(x)$. With this one defines the operator-valued column vector $b=(b_1, \ldots, b_{|\Lambda|})^T$ as
\begin{equation} \label{bfromqp}
b = \frac{1}{\sqrt{2}} (\gamma^{1/2} {\mathcal O}^T q + i \gamma^{-1/2} {\mathcal O}^T p)
\end{equation}
or, in more detail,
\begin{equation} \label{bfromqp2}
b_j = \frac{1}{\sqrt{2}} \sum_{x\in \Lambda} \varphi_j(x) (\gamma_j^{1/2} q_x +i\gamma_j^{-1/2} p_x), \quad 1 \le j \le |\Lambda|.
\end{equation}
These operators and their adjoints satisfy the canonical commutation relations (CCR)
\begin{equation}\label{b:CCR}
[b_j,b_m]=[b_j^*,b_m^*]=0, \text{ and }[b_j,b_m^*]=\delta_{j,m}\ \idty, \:\: 1\le j,m \le |\Lambda|
\end{equation}
and allow to rewrite $H_\Lambda$ as a free boson system \cite{NSS1},
\begin{equation}
H_\Lambda=\sum_{j} \gamma_j(2b_j^* b_j+\idty).
\end{equation}
This means that a complete set of eigenvectors $\psi_\alpha$, $\alpha=(\alpha_1, \ldots, \alpha_{|\Lambda|})\in\mathbb{N}_0^{|\Lambda|}$, of $H_\Lambda$ is given as
\begin{equation} \label{eigenvectors}
\psi_\alpha=\prod_{j=1}^{|\Lambda|} \frac{1}{\sqrt{\alpha_j!}}(b_j^*)^{\alpha_j}\psi_0,
\end{equation}
where  $\psi_0$ is the non-degenerate normalized ground state of $H_\Lambda$, characterized by $b_j \psi_0=0$ for all $j$, and $\alpha\in\mathbb{N}_0^{|\Lambda|}$ is called the occupation number vector. The corresponding eigenvalues of $H_\Lambda$ are
\begin{equation} \label{eigenvalues}
E_\alpha=\sum_{j=1}^{|\Lambda|} \gamma_j(2\alpha_j+1).
\end{equation}

Continuity of the one-body eigenvalues $\gamma_j$ in $k$ implies continuity of the many-body eigenvalues $E_\alpha$ in $k$. One can also check that the many-body eigenvectors $\psi_\alpha$ in $\mathcal{H}_\Lambda$ are weakly continuous in $k$ if one chooses the unique {\it positive} ground state of $H_\Lambda$ as vacuum vector $\psi_0$. (For $\psi_0$ this can be seen from its explicit characterization as $\psi_0(x) = \phi_0(\mathcal{O}^Tx)$ with $\phi_0(y) = \prod_j (\gamma_j/\pi)^{1/4} e^{-\gamma_j y_j^2/2}$. This can then be ``lifted'' to the excited states $\psi_\alpha$ via \eqref{eigenvectors}, using that the operators $b_j$ are weakly continuous on, say, the Schwartz space functions, due to \eqref{bfromqp}.) In particular, this gives $\mathbb{P}$-measurability of the $E_\alpha$ and $\psi_\alpha$ and all the sets and functions which will be relevant for our analysis below.

Next, we introduce the subspace of the many-body Hilbert space $\mathcal{H}_\Lambda$ which will represent the many-body localized regime identified by all our main results. Let
\begin{equation} \label{restrict1}
S_{\lambda_0}:=\{j\in \{1,\ldots,|\Lambda|\};\ \gamma_j^2\in[0,\lambda_0]\}
\end{equation} and
\begin{equation} \label{restrict2}
\I: = \{ \alpha\in\mathbb{N}_0^{|\Lambda|}: \; \supp \alpha \subset S_{\lambda_0}\},
\end{equation}
meaning the set of all $\alpha\in\mathbb{N}_0^{|\Lambda|}$ with $\alpha_j=0$ for $j \notin S_{\lambda_0}$ .

We will show many-body localization for $H_\Lambda$ on the subspace of $\mathcal{H}_\Lambda$ spanned by the eigenvectors $\psi_\alpha$ with $\alpha \in \I$, i.e., the excitations of the ground state $\psi_0$ corresponding to localized states of the effective Hamiltonian $h_\Lambda$. Thus we will refer to the range of the spectral projection
\begin{equation} \label{restrict3}
P_\I :=P_\I (H_\Lambda)=\sum_{\alpha\in {\I}}|\psi_\alpha\rangle\langle\psi_\alpha|
\end{equation}
as the {\it regime of localized excitations} for $H_\Lambda$. Due to almost sure non-degeneracity of $h_\Lambda$ and corresponding uniqueness of the eigenfunction basis $\{\varphi_j\}$, we see from \eqref{bfromqp2} and \eqref{eigenvectors} that the projections $|\psi_\alpha \rangle \langle \psi_\alpha|$ and thus $P_{\I}$ are almost surely uniquely determined by $H_\Lambda$.

Note that, while by \eqref{eigenvalues} the range of $P_{\I}$ includes all eigenstates of $H_\Lambda$ to energies in $[E_0, E_0+2\sqrt{\lambda_0}]$, the localized excitations are not merely a low-energy regime for the many-body Hamiltonian. In fact, the range of $P_\I$ contains states with positive many-body energy density. In Appendix~\ref{App:energydensity} we will give a more precise expression for this energy density in terms of the density of states of the infinite volume Anderson model.

The range of many-body energies covered by the regime of localized excitations may generally be a mixed regime of (many-body) localized and delocalized states, due to the unproven possibility of extended states above $\lambda_0$ in the Anderson model. However, at sufficiently high disorder of the distribution $\mu$ one has full localization of the Anderson model and thus can choose $\lambda_0=\infty$ and $P_{\I}= \idty$. For this case Theorem~\ref{thm:Weyl} below was proven in \cite{NSS1} (Theorem~3.3). Theorems~\ref{thm:quasiloc} and \ref{thm:Weyl-Eigen-ExpDecay} are new, also for this case, and apply to states in the entire many-body Hilbert space ${\mathcal H}_\Lambda$, but reflect a dependance of constants on the maximal occupation number $\|\alpha\|_\infty := \max \{|\alpha_j|: j=1,\ldots,|\Lambda|\}$ of excitation vectors $\alpha \in \I$.

In our results below we will describe dynamical localization of the many-body system $H_\Lambda$ through the change of the support of local observables $A$ under the Heisenberg evolution $\tau_t(A) = e^{itH_\Lambda}Ae^{-itH_\Lambda}$. The restriction of our results to the regime of localized excitations will be reflected through two-sided projection into this regime, i.e, we will consider
\begin{equation} \label{Irestrict}
 A_\I:=P_\I A P_\I,
\end{equation}
and the restricted Heisenberg evolution $\tau_t(A_\I )=\tau_t(A)_\I$ for suitable observables $A$.

%
%

\section{Preliminaries} \label{sec:prelim}

\subsection{Weyl operators and their restrictions} \label{sec:Weyl}

As the local Hilbert space $\mathcal{L}^2(\R)$ in \eqref{Hilbertspace} is infinite-dimensional, local observables may be unbounded. A convenient class of bounded observables which generates an irreducible sub-algebra of $\mathcal{B}({\mathcal H}_\Lambda)$ and for which we will state all our results is given by the Weyl (or \textit{displacement}) operators. For $f:\Lambda\rightarrow \mathbb{C}$ these are defined as the unitary operators
\begin{equation}
\W(f)=\exp{\left(i(q(f)+p(f))\right)},
\end{equation}
with the position and momentum operators
\begin{equation} \label{posmomops}
q(f)=\sum_{x\in\Lambda} \RE[f(x)] q_x, \quad p(f)=\sum_{x\in\Lambda} \IM[f(x)] p_x.
\end{equation}
While the latter are unbounded, it is well known that they have sufficiently large sets of analytic vectors, so that all formal manipulations used below are justified.
Note that $\W^*(f)=\W(-f)$  and $\W(0)=\idty$. It is also clear that
$\supp(\W(f))$ (in the sense of a the support of an operator on the tensor product \eqref{Hilbertspace}) coincides with $\supp(f)$ (the support of the function $f$).

We start by restating two basic properties of the Weyl operators, e.g.\ \cite{Brat-Rob2, NSS1}.

(i) The Weyl operators satisfy the so-called Weyl relations, i.e., for any $f,g:\Lambda\rightarrow\mathbb{C}$,
\begin{equation}\label{eq:Weyl-Relations}
\W(f+g)=e^{\frac{i}{2}\IM[\langle f,g\rangle]}\W(f)\W(g) = e^{-\frac{i}{2}\IM[\langle f,g\rangle]}\W(g)\W(f).
\end{equation}

(ii) The Heisenberg dynamics under $H_\Lambda$ of the Weyl operators $\tau_t(\W(f)) = e^{itH_\Lambda} \W(f) e^{-itH_\Lambda}$ is given by the following formula, which quantifies the fact that Weyl operators are mapped to Weyl operators under the time evolution:
\begin{equation} \label{Weyldyn}
\tau_t(\W(f))=\W(f_t), \text{ where } f_t=V^{-1}e^{2it\gamma}V f.
\end{equation}
Here the real-linear $V:\C^\Lambda \to \C^{|\Lambda|}$ is defined as
\begin{equation}\label{def:Weyl:Vf}
V f = \gamma^{-1/2}\mathcal{O}^T\RE[f] + i\gamma^{1/2}\mathcal{O}^T \IM[f],
\end{equation}
which is invertible with inverse
\begin{equation} \label{Vinverse}
V^{-1}g = \mathcal{O} \gamma^{1/2} \RE[g] +i\mathcal{O} \gamma^{-1/2} \IM[g].
\end{equation}

Additional properties of the Weyl operators which we will need are provided with proofs in Appendix~\ref{sec:Weylops}.

Next, we need to understand the restrictions $\W(f)_\I = P_\I \W(f)P_\I$ of the Weyl operators to the reducing subspaces introduced in \eqref{restrict1} to \eqref{restrict3} above. For this we will use that the Weyl operators can be expressed in terms of the operators $b_j$ from \eqref{bfromqp} as
\begin{equation} \label{defWeyl}
\W(f)=\exp\left(\frac{i}{\sqrt{2}}(b(Vf)+b^*(Vf))\right).
\end{equation}
Here $b(g) := \sum_j \bar{g}_j b_j$ and $b^*(g) := \sum_j g_j b_j$, from which $q(f)+p(f) = \frac{1}{\sqrt{2}}( b(Vf)+b^*(Vf))$ is found by a simple calculation.

Throughout the following we will write $\X:=\chi_{[0,\lambda_0]}(h_\Lambda) = \mathcal{O} \idty_{S_{\lambda_0}} \mathcal{O}^T$ for the orthogonal projection onto the localized energy regime of the effective Hamiltonian $h_\Lambda$.
In the proof of the following lemma we will also use the alternative representation
\begin{equation} \label{eq:Vinv-Chi-V}
\X = V^{-1} \idty_{S_{\lambda_0}} V,
\end{equation}
which easily follows from \eqref{def:Weyl:Vf} and \eqref{Vinverse}.

\begin{lem}[Restriction of Weyl Operators] \label{lem:Weyl-restriction}
For $f:\Lambda\rightarrow\C$,
\begin{equation}\label{eq:PIWPI}
\W(f)_\I  = C_f \W(\X f)P_\I,
\end{equation}
where $C_f :=\exp\left(-\frac{1}{4}\|\idty_{S_{\lambda_0}^c}Vf\|^2\right)$. Moreover, we have
\begin{equation}\label{eq:PI:commutewith-W}
[P_\I ,\W(\X f)]=0.
\end{equation}
\end{lem}

Note that, since $\W(\cdot)$ is unitary and $P_\I $ is a non-zero orthogonal projection (as $\psi_0 \in \Ran(P_\I)$), (\ref{eq:PIWPI}) gives
\begin{equation}
0<\|\W(f)_\I \|=C_f \le 1
\end{equation}
and $\|\W(f)_\I \|<1$ if supp$\,Vf \not\subset S_{\lambda_0}$.

\begin{proof}

The proof of this lemma follows from Lemma~\ref{thm:Weyl:entries}. In fact, using e.g.\ (\ref{eq:Weyl:product}), one readily checks that
the quantity $\W( \X f)$ is diagonal with respect to $P_{\mathcal{I}}$, i.e.
\begin{equation}
P_\I \W(\X f) (1 - P_{\I}) = 0 = (1 - P_{\I}) \W(\X f) P_{\I}.
\end{equation}
As a result, (\ref{eq:PI:commutewith-W}) is clear.

If we now denote by $\X^c = V^{-1} \idty_{S^c_{\lambda_0}} V$, one has that
\begin{equation}
\W(f) = \W(\X^c f) \W(\X f)
\end{equation}
and then (\ref{eq:PIWPI}) is again a consequence of (\ref{eq:Weyl:product}); here we use specifically (\ref{fact1}).
\end{proof}

\subsection{Localization of the effective dynamics} \label{sec:effdyn}

The Heisenberg dynamics of the Weyl operators is related to the effective dynamics $f_t=V^{-1}e^{2it\gamma}V f$ on $\ell^2(\Lambda)$ through \eqref{Weyldyn}. A crucial link between one-body and many-body localization properties will thus be given by expressing the one-body localization bound \eqref{def:eig-corr} in terms of a localization bound for the effective dynamics.

To state this, it is useful to identify
\begin{equation}
f \in \ell^2( \Lambda; \mathbb{C}) \quad \mbox{with} \quad \begin{pmatrix} {\rm Re}[f] \\ {\rm Im}[f] \end{pmatrix} \in \ell^2( \Lambda; \mathbb{R}) \oplus \ell^2( \Lambda; \mathbb{R}) \, .
\end{equation}
In particular, for any $f,g \in \ell^2( \Lambda; \mathbb{C})$ it is clear that
\begin{equation} \label{real-ip}
{\rm Re}[ \langle f,g \rangle] = \left\langle \begin{pmatrix} {\rm Re}[f] \\ {\rm Im}[f] \end{pmatrix}, \begin{pmatrix} \idty & 0 \\ 0 & - \idty \end{pmatrix} \begin{pmatrix} {\rm Re}[g] \\ {\rm Im}[g] \end{pmatrix} \right\rangle
\end{equation}
and similarly
\begin{equation} \label{imag-ip}
{\rm Im}[ \langle f,g \rangle] = \left\langle \begin{pmatrix} {\rm Re}[f] \\ {\rm Im}[f] \end{pmatrix}, \begin{pmatrix} 0 & \idty \\  - \idty  & 0 \end{pmatrix} \begin{pmatrix} {\rm Re}[g] \\ {\rm Im}[g] \end{pmatrix} \right\rangle .
\end{equation}

\begin{lem}[Localization of the Effective Dynamics] \label{lem:ave-dec-l2} Let $f,g \in \ell^2( \Lambda; \mathbb{C})$. For any $t \in \mathbb{R}$, take $g_t = V^{-1} e^{2i t \gamma} Vg$;
see (\ref{Weyldyn}). Under the assumption of eigencorrelator decay, i.e. (\ref{def:eig-corr}), one has that
\begin{equation} \label{ave-dyn-in-l2}
\mathbb{E} \left( \sup_{t \in \mathbb{R}} | \langle f, \X g_t \rangle | \right) \leq 2 C (1+ \lambda_0^{1/2})^2 \sum_{x,y \in \Lambda} |f(x)| |g(y)| e^{- \mu|x-y|}
\end{equation}
Here $C$ and $\mu$ are as in (\ref{def:eig-corr}).
\end{lem}

\begin{proof}
First, note that absorbing an extra term $h_\Lambda^{1/2}$ and $h_\Lambda$, respectively, into $u(h_\Lambda)$, one sees that \eqref{def:eig-corr} implies the bounds
\begin{equation} \label{eig-corr-2}
\mathbb{E}\left(\sup_{|u|\leq 1} \left|\langle\delta_x, u(h_\Lambda) \X \delta_y\rangle \right|\right)  \le C \lambda_0^{1/2} e^{-\mu|x-y|},
\end{equation}
\begin{equation} \label{eig-corr-3}
\mathbb{E}\left(\sup_{|u|\leq 1} \left|\langle\delta_x,h_{\Lambda}^{1/2} u(h_\Lambda)\X\delta_y\rangle \right|\right)  \le  C \lambda_0 e^{-\mu|x-y|}.
\end{equation}

We prove (\ref{ave-dyn-in-l2}) by estimating the real and imaginary parts of the inner-product separately.
In fact, a short calculation based on \eqref{def:Weyl:Vf} and \eqref{Vinverse} and the fact that $\X$ is a real operator shows that
\begin{equation}
\begin{pmatrix} {\rm Re}[ \X g_t] \\ {\rm Im}[ \X g_t] \end{pmatrix} = \begin{pmatrix} \cos(2t h_{\Lambda}^{1/2}) \X & - \sin(2th_{\Lambda}^{1/2}) h_{\Lambda}^{1/2} \X \\  \sin(2th_{\Lambda}^{1/2}) h_{\Lambda}^{-1/2} \X & \cos(2t h_{\Lambda}^{1/2}) \X \end{pmatrix} \begin{pmatrix} {\rm Re}[g] \\ {\rm Im}[g] \end{pmatrix}.
\end{equation}
In this case, a rough estimate, using (\ref{real-ip}), implies
\begin{eqnarray}
| {\rm Re}[ \langle f, \X g_t \rangle ] | & \leq & \sum_{x,y \in \Lambda} |f(x)| |g(y)| \left( \langle \delta_x, \cos(2t h_{\Lambda}^{1/2}) \X \delta_y \rangle +  \langle \delta_x, h_{\Lambda}^{1/2} \sin(2 th_{\Lambda}) \X \delta_y \rangle \right) \nonumber \\
& \mbox{ } & \quad +\sum_{x,y \in \Lambda}  |f(x)| |g(y)| \left( \langle \delta_x, h_{\Lambda}^{-1/2} \sin(2th_{\Lambda}^{1/2}) \X \delta_y \rangle +  \langle \delta_x, \cos(2th_{\Lambda}^{1/2}) \X \delta_y \rangle  \right)
\end{eqnarray}
A similar estimate, using (\ref{imag-ip}), applies to $| {\rm Im}[ \langle f, \X g_t \rangle ] |$.
The result in (\ref{ave-dyn-in-l2}) now follows immediately from an application of (\ref{def:eig-corr}), (\ref{eig-corr-2}) and (\ref{eig-corr-3}).
\end{proof}

Bounds as the sum on the right hand side of \eqref{ave-dyn-in-l2} naturally appear in our proofs, also in several of our main results below. We will generally keep the bounds in this form, but one can also state them as bounds involving $\ell^p$-norms of $f$ and $g$. Most directly, one gets exponential decay in the distance of $\supp f$ and $\supp g$, with constants proportional to $\|f\|_1$ and $\|g\|_1$.  Using H\"older's inequality (and saving part of the factor $e^{-\mu|x-y|}$ for decay in the distance of supports), one can also turn this into bounds in terms of other $\ell^p$-norms, if more suitable for a desired application.

%
%

\section{Lieb-Robinson bounds} \label{sec:LRBs}

We now turn to our first main result and its proof, a zero-velocity Lieb-Robinson bounds for the Heisenberg dynamics of Weyl operators as well as of local position and momentum operators, in each case restricted via the projection $P_\I$ onto the regime of localized excitations.

To conveniently state our results for position and momentum operators,  we introduce the $2\times 2$-block matrix
\begin{equation}
A_{t,\I}(f,g):=\begin{pmatrix}
  [\tau_t(q(f)_\I),q(g)_\I] & [\tau_t( q(f)_\I),p(g)_\I] \\
  [\tau_t(p(f)_\I),q(g)_\I] & [\tau_t(p(f)_\I),p(g)_\I] \\
\end{pmatrix}.
\end{equation}
Thus for the four choices $j, k \in \{1,2\}$ the matrix elements $(A_{t,\I}(f,g))_{j,k}$ cover all possible Lieb-Robinson-type commutators between local position operators $q(f)$, $q(g)$ and local momentum operators $p(f)$, $p(g)$ as defined by \eqref{posmomops}.

\begin{thm}[Restricted Lieb-Robinson Bounds]\label{thm:Weyl}
For any $f,g: \Lambda \to \C$,
\begin{equation} \label{WeylLR}
\mathbb{E}\left(\sup_{t\in\mathbb{R}}\left\|\left[\tau_t(\W(f)_\I ),\W(g)_\I \right]\right\|\right)\leq
C(1+\lambda_0^{1/2})^2\sum_{x,y \in \Lambda} |f(x)| |g(y)|\ e^{-\mu|x-y|}.
\end{equation}

Moreover, for all $j,k\in\{1,2\}$,
\begin{equation} \label{posmomLR}
\mathbb{E}\left(\sup_{t\in\mathbb{R}}\|(A_{t,\I}(f,g))_{j,k}\|\right)\leq
C\lambda_0^{\frac{j+k-2}{2}}\sum_{x,y \in \Lambda}|f(x)| |g(y)|\ e^{-\mu|x-y|}.
\end{equation}
Here  $C$ and $\mu$ are the constants in the eigencorrelator localization bound (\ref{def:eig-corr}).
\end{thm}

As discussed at the end of Section~\ref{sec:effdyn}, this gives disorder averaged Lieb-Robinson bounds, exponentially decaying in the distance of the supports of the observables considered, uniform in time, thus interpreted as zero-velocity LR bounds. For the case of large disorder, where \eqref{def:eig-corr} holds with $\lambda_0=\infty$ and restriction to the range of $P_\I$ is not necessary, this has been shown previously in \cite{NSS1}.

While we restrict to the regime of localized excitations, we point out that the bounds do not depend on the maximal excitation number $\|\alpha\|_\infty$ of the states $\varphi_\alpha$ spanning $P_\I$ (in other words, we do not have to restrict with projections smaller than $P_\I$ which would limit the excitation number). This is different, and in some sense stronger, than the situation we will find for the quantities considered in Theorems~\ref{thm:quasiloc} and \ref{thm:Weyl-Eigen-ExpDecay} below.

\begin{proof}
For the proof of \eqref{WeylLR} we start from
\begin{eqnarray}
\left[\tau_t\left(\W(f)_\I \right),\W(g)_\I \right]&=&\left[\W(f_t)_\I ,\W(g)_\I \right] \\
&=& C_{f_t} C_g \left[\W(\X f_t), \notag
\W(\X g)\right]P_\I
\end{eqnarray}
where we used \eqref{eq:PIWPI} and (\ref{eq:PI:commutewith-W}). Then with $\|P_\I \|= 1$, $C_{f_t} C_g \le 1$ and the Weyl relations (\ref{eq:Weyl-Relations}),
we find
\begin{eqnarray}\label{eq:LR:proof:step1}
\left\|\left[\tau_t(\W(f)_\I ),\W(g)_\I \right]\right\|&\leq& \left\|[\W(\X f_t),\W(\X g)]\right\| \\
&=& \left|(e^{-i\IM[\langle\X f_t,g\rangle]}-1)\right|\left\|\W(\X g)\W(\X f_t )\right\| \notag \\
&\leq& |\IM[\langle\X f_t,g\rangle]|. \notag
\end{eqnarray}
The claim in (\ref{WeylLR}) now follows from an application of Lemma~\ref{lem:ave-dec-l2}. Note that we saved a factor of
2 since we need only consider the imaginary part above.

Our proof of (\ref{posmomLR}) mimics closely the methods in Section 3.2 of \cite{NSS1}. In fact, for any $x,y \in \Lambda$ and all $t \in \mathbb{R}$,
let us define
\begin{equation}
A_{t,\I}(\delta_x,\delta_y) = -i \begin{pmatrix} [ \tau_t( (q_x)_{\I}), (q_y)_{\I}] &  [ \tau_t( (q_x)_{\I}), (p_y)_{\I}] \\
 [ \tau_t( (p_x)_{\I}), (q_y)_{\I}] &  [ \tau_t( (p_x)_{\I}), (p_y)_{\I}]  \end{pmatrix} \, ,
\end{equation}
compare with (3.27) in \cite{NSS1}. Arguing as in the proof of Lemma 3.7 in \cite{NSS1}, a short calculation shows that
\begin{equation} \label{p+q:com}
A_{t,\I}(\delta_x,\delta_y) =  \begin{pmatrix} - \langle \delta_x, h_{\Lambda}^{-1/2} \sin(2 t h_{\Lambda}^{1/2}) \X \delta_y \rangle P_{\I} &   \langle \delta_x,  \cos(2 t h_{\Lambda}^{1/2}) \X \delta_y \rangle P_{\I} \\
 - \langle \delta_x, \cos(2 t h_{\Lambda}^{1/2}) \X \delta_y \rangle P_{\I} & - \langle \delta_x, h_{\Lambda}^{1/2} \sin(2 t h_{\Lambda}^{1/2}) \X \delta_y \rangle P_{\I}  \end{pmatrix}
\end{equation}
For this calculation, we observe that the relation
\begin{equation}
\begin{pmatrix}
  b_\I \\
  b^*_\I \\
\end{pmatrix}
\left( (b_\I)^T, (b^*_\I)^T \right)
-\left(
\begin{pmatrix}
  b_\I \\
  b^*_\I \\
\end{pmatrix}
\left( (b_\I)^T, (b^*_\I)^T \right)
\right)^T=\begin{pmatrix}
            0 & P_\I  \otimes \idty_{S_{\lambda_0}} \\
            -P_\I  \otimes \idty_{S_{\lambda_0}} & 0 \\
          \end{pmatrix}.
\end{equation}
replaces (3.31) in the proof of Lemma 3.7 of \cite{NSS1}. Here $b_{\I}$, and similarly $b^*_{\I}$, is the
$| \Lambda|$-dimensional column vector with operator-valued entries given by
$(b_{\I})_j = P_{\I}b_jP_{\I}$ for $1 \leq j \leq |\Lambda|$. Moreover,
$P_\I  \otimes \idty_{S_{\lambda_0}}$ is our notation for the $| \Lambda| \times | \Lambda|$
diagonal matrix with operator-valued entries given by $(P_\I  \otimes \idty_{S_{\lambda_0}})_{jj} = P_{\I}$
if $j \in S_{\lambda_0}$ and $0$ otherwise.

Expanding $f$ and $g$ in the basis $\{ \delta_x \}_{x \in \Lambda}$, it is clear that (\ref{posmomLR})
follows from (\ref{p+q:com}) and our eigenfunction correlator assumption (\ref{def:eig-corr}); see also
(\ref{eig-corr-2}) and (\ref{eig-corr-3}).

\end{proof}

%
%

\section{On quasi-locality estimates for restricted Weyl operators} \label{sec:quasiloc}

In this section, we will prove a quasi-locality estimate for Weyl operators restricted to
the regime of localized excitations.  
We begin with a brief description of this result which is the content of
Theorem~\ref{thm:quasiloc} below.
Recall that, as has been discussed in Section~\ref{sec:Weyl}, the harmonic evolution of Weyl operators satisfies the following relation:
for any $f : \Lambda \to \mathbb{C}$
\begin{equation} \label{harm-dyn-weyl}
\tau_t(\W(f)) = \W(f_t) \quad \mbox{where} \quad f_t = V^{-1} e^{2it \gamma} V f \, ,
\end{equation}
see (\ref{Weyldyn}), with $V$ being the operator defined in (\ref{def:Weyl:Vf}). If $f$ is a strictly local function, i.e. the
support of $f$ is contained in some $X \subset \Lambda$, then it is clear that the
corresponding Weyl operator is strictly local as well, in the sense that
$\W(f)$ is an element of the sub-algebra $\mathcal{B}( \mathcal{H}_X) \subset \mathcal{B}( \mathcal{H}_{\Lambda})$. Due to non-relativistic effects, strict
locality is not preserved by the harmonic dynamics; if $f$ has support in $X \subset \Lambda$,
then generically, for any $t \neq 0$, $f_t$ and therefore $\tau_t(\W(f))$ have non-trivial support
on all of $\Lambda$. The goal of this section is to prove a quasi-locality estimate
for these dynamically evolved Weyl operators. In particular,
we will provide an estimate on the difference between the full dynamics and
a strictly local approximation. Moreover, disorder allows us to demonstrate that our
bounds hold uniformly in time.

The basic set-up is as follows.

Let $X \subset \Lambda$ and $f: \Lambda \to \mathbb{C}$ be a function with support in $X$;
we will denote this by ${\rm supp}(f) \subset X$. To define our strictly local approximation,
it is convenient to introduce
\begin{equation}\label{def:Xn}
X(n) =\{x\in\Lambda;\ \text{dist}(x,X)\leq n\}
\end{equation}
for any $n \geq 0$. We may refer to $X(n)$ as the $n$-neighborhood of $X$.
Let us also introduce a notion of boundary for sets $X \subset \Lambda$; namely set
\begin{equation}
\partial X = \{ x \in X : \mbox{ there exists } y \in \Lambda \setminus X \mbox{ with } |x-y| =1 \} \, .
\end{equation}

Since we have not assumed full many-body localization, our results require an energy restriction to the
regime of localized excitations.
An application of Lemma~\ref{lem:Weyl-restriction} shows that
 \begin{equation} \label{e-proj-weyl}
 \tau_t(\W(f))_{\mathcal{I}} = \W(f_t)_{\mathcal{I}} = C_{f_t} \W( \X f_t) P_{\mathcal{I}}
 \end{equation}
 where the positive prefactor $C_{f_t} = C_f$ is independent of time. In fact, from
 the form of $f_t$, see (\ref{harm-dyn-weyl}), it is clear that $\| \idty_{S_{\lambda_0}^c} Vf_t \|  = \| \idty_{S_{\lambda_0}^c} Vf \|$.
Based on (\ref{e-proj-weyl}), we will choose a strictly local approximation of
 $\tau_t(\W(f))$ proportional to $\W( \idty_{X(n)} \X f_t)$; the
latter trivially having support in $X(n) \subset \Lambda$. More precisely, we set 
\begin{equation} \label{strictly-local-app}
\widehat{W} = \widehat{C} \cdot \W( \idty_{X(n)} \X f_t) \quad \mbox{and choose} \quad \widehat{C} =   \frac{C_f}{C_{\idty_{X(n)} \X f_t}} \, .
\end{equation}
The positive factors $C_f$ and $C_{\idty_{X(n)} \X f_t}$ are as in the statement of Lemma~\ref{lem:Weyl-restriction}. In this case, it is clear that $\|\widehat{W}_{\I}\|=C_f\leq 1$.

As is well-known, for any $f \neq 0$, $\| \W(f) - \idty \| = 2$, see e.g. \cite{Brat-Rob2}[Proposition 5.2.4 (5)], and
so we estimate in the strong operator topology. More precisely, let $\kappa \geq 0$ be an integer and
\begin{equation} \label{Hkappa}
\mathcal{H}^{(\kappa)} = \mbox{span}\{\psi_\alpha: \|\alpha\|_\infty \le \kappa \}.
\end{equation}
This is a reducing subspace for $H_\Lambda$ and we will write $\mathcal{D}^{(\kappa)}$ for the set of all normalized eigenvectors of $H_\Lambda$ in $\mathcal{H}^{(\kappa)}$. Note that by \eqref{simplespec} we have almost surely that {\it all} $\psi \in \mathcal{D}^{(\kappa)}$ are, up to a phase, of the form $\psi_\alpha$, $\|\alpha\|_\infty \le \kappa$ (and no additional eigenvectors are found by linear combination).

\begin{thm}[Quasi-Locality for Restricted Weyl Operators] \label{thm:quasiloc}
Let $X \subset \Lambda$ and $f : \Lambda \to \mathbb{C}$ satisfy ${\rm supp}(f) \subset X$.
Under assumption (\ref{def:eig-corr}), for any $\kappa \geq 0$ and $n \geq 0$, there is $\tilde{C}>0$
for which the bound
\begin{equation}\label{eq:quasi-locality}
\mathbb{E}\left(\sup_{\psi \in \mathcal{D}^{(\kappa)}} \sup_{t\in\mathbb{R}}\left\|\left(\tau_t(\W(f))-
\widehat{W} \right)_{\I}\psi \right\|\right)\leq \tilde{C}(1+ \kappa)^{1/3} (1+ \lambda_0^{1/2})^{4/3} | \partial X | \| f \|_{\infty}^{2/3} e^{- \mu n /3}
\end{equation}
holds. Here one may take
\begin{equation}
\tilde{C} = 2^{4/3} C \left( \sum_{z \in \mathbb{Z}^\nu} e^{- \mu |z|/6} \right)^4
\end{equation}
and we note that $C$ and $\mu$ are the constants in the eigencorrelator localization bound (\ref{def:eig-corr}).
\end{thm}

The key statement in this result is the exponential decay of the right hand side of \eqref{eq:quasi-locality} in $n$, uniformly in time. Note that such quasi-locality bounds have not been established in the context of disordered oscillator systems before and that Theorem~\ref{thm:quasiloc} is new even for the fully localized case where $P_\I=\idty$ and one can choose $\lambda_0 = 4\nu + \|k\|_{\infty}$ as a bound on $\|h_\Lambda\|$.

\begin{proof}[Proof of Theorem~\ref{thm:quasiloc}]
The first step in our proof of (\ref{eq:quasi-locality}) involves a simple norm bound.
In fact, for any self-adjoint operator $A$ and each $\psi$ in its domain, the bound
\begin{equation}
\| (e^{iA} - \idty) \psi \| = \left\| \int_0^1 \left( \frac{d}{ds} e^{isA} \right) \, \psi \, ds \right\| \leq \| A \psi \|
\end{equation}
is clear. For our application to Weyl operators, take $\psi_{\alpha}$, as in (\ref{eigenvectors}), and note that for any
$f : \Lambda \to \mathbb{C}$ one has that
\begin{eqnarray} \label{Weyl_dif_bd}
\| ( \W(f) - \idty) \psi_{\alpha} \| & \leq & \frac{1}{\sqrt{2}} \left( \| b(Vf) \psi_{\alpha} \| + \| b^*(Vf) \psi_{\alpha} \| \right) \\
& \leq & \sqrt{2 (\| \alpha \|_{\infty} +1) } \| Vf \|_2 \notag.
\end{eqnarray}
The final bound above uses e.g.\ that
\begin{equation}
\| b^*(Vf) \psi_{\alpha} \|^2 = \left\| \sum_{j=1}^{| \Lambda|} (Vf)(j) \sqrt{\alpha_j+1} \psi_{\alpha + e_j} \right\|^2 \leq (\| \alpha \|_{\infty} +1) \| Vf \|_2^2 \, ,
\end{equation}
valid since the collection $\{ \psi_{\alpha + e_j} \}$ is orthonormal, and a similar bound for $\| b(Vf) \psi_{\alpha} \|$.

For $\widehat{W}$ as defined in (\ref{strictly-local-app}), it is clear that
\begin{eqnarray} \label{sl-app}
( \tau_t(\W(f)) - \widehat{W})_{\mathcal{I}} & = & C_{f_t} \W( \X f_t) P_{\I} - \widehat{C} \cdot C_{\idty_{X(n)} \X f_t} \W( \X \idty_{X(n)} \X f_t) P_{\I} \\
& = &  C_f \W( \X \idty_{X(n)} \X f_t) \left( \W( -\X \idty_{X(n)} \X f_t)  \W( \X f_t) - \idty \right)P_{\I} \nonumber \\
& = &  C_f \W( \X \idty_{X(n)} \X f_t) \left( \W(f_{n,t}) - \idty \right)P_{\I} \nonumber
\end{eqnarray}
where, for the last line above, we used the Weyl relations (\ref{eq:Weyl-Relations}) and 
set $f_{n,t} = \X \idty_{\Lambda \setminus X(n)} \X f_t$ (note that $\langle -\X \idty_{X(n)} \X f_t, \X f_t\rangle$ is real, so that no phase appears). The bound in (\ref{Weyl_dif_bd}) then immediately yields
that
\begin{equation} \label{ql-est-1}
\| ( \tau_t(\W(f)) - \widehat{W})_{\mathcal{I}} \psi_{\alpha} \| \leq \sqrt{2( \| \alpha \|_{\infty} +1)} \| Vf_{n,t} \|_2
\end{equation}
for any $\psi_{\alpha}$ in the range of $P_{\mathcal{I}}$.

Now a short calculation, using that $V \X = \idty_{S_{\lambda_0}}V$, see e.g. (\ref{eq:Vinv-Chi-V}), shows that
\begin{eqnarray} \label{exp-l2-norm}
\| V f_{n,t} \|_2^2 & = & \langle V \X \idty_{\Lambda \setminus X(n)} \X f_t, V \X \idty_{\Lambda \setminus X(n)} \X f_t \rangle \nonumber \\
& = & \sum_{z,z' \in \Lambda \setminus X(n)} \langle \X f_t, \delta_z \rangle \langle V \delta_z, \idty_{S_{\lambda_0}} V \delta_{z'} \rangle \langle \delta_{z'}, \X f_t \rangle
\end{eqnarray}
In this case, we conclude that
\begin{eqnarray} \label{ql-est-2}
\| ( \tau_t(\W(f)) - \widehat{W})_{\mathcal{I}} \psi_{\alpha} \| & \leq & 2^{1/3} \| ( \tau_t(\W(f)) - \widehat{W})_{\mathcal{I}} \psi_{\alpha} \|^{2/3} \nonumber \\
& \leq & 4^{1/3} ( \| \alpha \|_{\infty} + 1)^{1/3} \| Vf_{n,t} \|_2^{2/3} \nonumber \\
& \leq & C_{\kappa} \sum_{z,z' \in \Lambda \setminus X(n)}
|\langle \X f_t, \delta_z \rangle|^{1/3} |\langle V \delta_z,  \idty_{S_{\lambda_0}} V \delta_{z'} \rangle|^{1/3} |\langle \delta_{z'}, \X f_t \rangle |^{1/3} \, .
\end{eqnarray}
Here we have used (\ref{ql-est-1}) and (\ref{exp-l2-norm}). In fact,
the right-hand-side above is uniform for $\psi_{\alpha} \in \mathcal{D}^{(\kappa)}$ and we have set
$C_{\kappa} = 4^{1/3} (\kappa + 1)^{1/3}$. An application of the H\"older inequality shows that
the left-hand-side of (\ref{eq:quasi-locality}) is bounded above by
\begin{equation} \label{ql-holder-bd}
C_{\kappa} \sum_{z,z' \in \Lambda \setminus X(n)} \mathbb{E} \left( \sup_{t \in \mathbb{R}} | \langle \X f_t, \delta_z \rangle | \right)^{1/3}
 \mathbb{E} \left(  |\langle V \delta_z,  \idty_{S_{\lambda_0}} V \delta_{z'} \rangle| \right)^{1/3}
 \mathbb{E} \left( \sup_{t \in \mathbb{R}} | \langle \delta_{z'}, \X f_t \rangle | \right)^{1/3}
\end{equation}

We are now in position to apply Lemma~\ref{lem:ave-dec-l2}. In fact, by \eqref{def:Weyl:Vf} and \eqref{diagh} we have
\begin{equation}
\langle V \delta_z, \idty_{S_{\lambda_0}} V \delta_{z'} \rangle = \langle \delta_z, h_{\Lambda}^{-1/2} \X \delta_{z'} \rangle
\end{equation}
and thus can bound the middle term above directly by \eqref{def:eig-corr}. We find the following upper bound on the quantity in (\ref{ql-holder-bd}) above, 
\begin{equation} \label{ql-lem-est}
2^{2/3} C_{\kappa} C (1+ \lambda_0^{1/2})^{4/3}  \sum_{z,z' \in \Lambda \setminus X(n)} \sum_{x,y} |f(x)|^{1/3} |f(y)|^{1/3} e^{- \mu|x-z|/3} e^{- \mu|z-z'|/3}  e^{- \mu|z'-y|/3}
\end{equation}

To obtain the bound claimed in (\ref{eq:quasi-locality}), we argue as follows. First, estimate the functions values by $\| f \|_{\infty}$.
Next, extract some decay in $n$, more precisely a factor of $e^{- \mu n/6}$, from those exponential terms involving $x$
as well as those involving $y$. Finally, sum on $y$; then sum on $z'$. This results in an upper bound of
\begin{equation}
2^{2/3} C_{\kappa} C (1+ \lambda_0^{1/2})^{4/3} \| f \|_{\infty}^{2/3} e^{- \mu n /3} \left( \sum_z e^{- \mu |z|/6} \right)^2 \cdot \sum_{x \in X} \sum_{z \in \Lambda \setminus X(n)} e^{- \mu |x-z|/6}
\end{equation}
for the quantity in (\ref{ql-lem-est}). Now, for each $x \in X$ and $z \in \Lambda \setminus X(n)$
the bound
\begin{equation}
e^{- \mu |x-z|/6}  \leq \sum_{w \in \partial X} e^{- \mu |x-w|/6} e^{- \mu |w-z|/6}
\end{equation}
is clear since equality is obtained for some $w \in \partial X$. The bound claimed in (\ref{eq:quasi-locality}) now follows.
\end{proof}

\section{Exponential decay of dynamic correlations of the Weyl operators} \label{sec:cordecay}

Our third main result establishes exponential decay of dynamic correlations of the Weyl operators in the regime of localized excitations, i.e., in the range of the projection $P_\I$ defined through \eqref{restrict1}, \eqref{restrict2} and \eqref{restrict3} above.

For a normalized $\psi \in \mathcal{H}_\Lambda$, any Weyl operators $\W(f)$ and $\W(g)$ for functions $f, g:\Lambda\rightarrow\C$, and time $t\in \R$, we define the $P_\I$-restricted dynamic correlation
\begin{equation} \label{def:cor}
C_\psi^{\I}(f,g,t) := \langle \psi, \tau_t(\W(f)_\I) \W(g)_\I \psi \rangle - \langle \psi, \tau_t(\W(f)_\I) \psi \rangle \langle \psi, \W(g)_\I \psi \rangle.
\end{equation}
For eigenstates $\psi$ of $\mathcal{H}_\Lambda$, the case considered below, we may simplify $\langle \psi, \tau_t(\W(f)_\I) \psi \rangle = \langle \psi, \W(f)_\I \psi\rangle$ on the right of \eqref{def:cor}.

As before, $\mathcal{D}^{(\kappa)}$ denotes the set of all normalized eigenvectors of $H_\Lambda$ in the subspace $\mathcal{H}^{(\kappa)}$ given by \eqref{Hkappa}, i.e., the localized excitations with occupation number bounded by $\kappa$.

\begin{thm}[Exponential Decay of Dynamic Correlations of the Weyl Operators] \label{thm:Weyl-Eigen-ExpDecay}
For all $\kappa \in \N_0$ and functions $f,g:\Lambda \to \C$,
\begin{equation}
\mathbb{E}\left(\sup_{\psi \in \mathcal{D}^{(\kappa)}} \sup_{t\in\mathbb{R}}|C_{\psi}^\I (f,g,t)|\right)\leq 8C(1+\lambda_0^{1/2})^2\big(\sum_{x,y\in\Lambda}|f(x)g(y)| e^{-\mu|x-y|}\big)^{\frac{1}{\kappa+1}}
\end{equation}
where  $C$ and $\mu$ are the constants in the eigencorrelator localization bound (\ref{def:eig-corr})
\end{thm}

We comment that in the fully localized large disorder regime, where $P_\I = \idty$, exponential decay of dynamic correlations of the Weyl operators and of local position and momentum operators was proven in \cite{NSS1}, but only for the ground state as well as for thermal states. For position and momentum operators and in the fully localized regime this was extended in \cite{ARSS17} to excited states. The latter work also considers the dynamics of correlations under a quantum quench, i.e., that the correlations remain exponentially decaying, uniformly in time, if the initial state is a product of either ground states or thermal states of subsystems.

In the remainder of this section we prove Theorem \ref{thm:Weyl-Eigen-ExpDecay}.

Due to simplicity \eqref{simplespec} of the spectrum of $H_\Lambda$, almost surely, all $\psi \in \mathcal{D}^{(\alpha)}$ are of the form $\psi_\alpha$ with $\|\alpha\|_\infty \le \kappa$. Thus we only need to study the special correlations $C^{\I}_\alpha(f,g,t) := C^{\I}_{\psi_\alpha}(f,g,t)$ and find a bound for
\begin{equation} \label{needtobound}
\E\left(\sup_{\alpha\,:\, \|\alpha\|_\infty \le \kappa} \sup_t \:C^{\I}_\alpha(f,g,t)\right).
\end{equation}

From (\ref{eq:PIWPI}) and (\ref{eq:PI:commutewith-W}) we get, using $C_{f_t} = C_f$ as remarked after \eqref{e-proj-weyl},
\begin{eqnarray} \label{corexpansion}
C^{\I}_\alpha(f,g,t)&=& C_f C_g\left(\langle \psi_\alpha, \W(\X f_t) \W(\X g) \psi_\alpha \rangle-\langle \psi_\alpha \W(\X f_t)\psi_\alpha \rangle \langle \psi_\alpha, \W(\X g) \psi_\alpha \rangle\right),\\
&=& C_{f}C_g\sum_{\beta\in\mathbb{N}_0^{|\Lambda|};\ \beta\neq \alpha} \langle\psi_\alpha, \W(\X f_t)\psi_\beta\rangle \langle\psi_\beta, \W(\X g)\psi_\alpha\rangle \nonumber\\
&=& C_{f}C_g\sum_{\beta\in\mathbb{N}_0^{|\Lambda|};\ \beta\neq \alpha}\ \ \ \prod_{\ell=1}^{|\Lambda|}\langle\alpha_\ell|\W_{\eta_{\ell,t}}|\beta_\ell\rangle \langle\beta_\ell|\W_{\xi_\ell}|\alpha_\ell\rangle, \notag \\
&=& C_{f}C_g\sum_{\beta\in\mathbb{N}_0^{|\Lambda|};\ \beta\neq \alpha}\ \ \ \prod_{\ell=1}^{|\Lambda|} W_{\alpha_\ell,\beta_\ell}^{\eta_{\ell,t}, \xi_\ell}, \notag
\end{eqnarray}
where the last step used the product formula (\ref{eq:Weyl:product})  for the Weyl operator expectations and we have set
\begin{equation} \label{etaxi}
\eta_{\ell,t}:=(V\X f_t)(\ell),\quad \xi_\ell:=(V\X g)(\ell),\text{ for }\ell\in\{1,\ldots,|\Lambda|\}
\end{equation}
as well as
\begin{equation}
W_{r,s}^{z,z'} := \langle r| \W_z|s \rangle \langle s| \W_{z'}|r \rangle.
\end{equation}
The summation over $\{\beta \in \N_0^{|\Lambda|}: \beta \not= \alpha\}$ in \eqref{corexpansion} can be decomposed into sums over the disjoint sets
\begin{eqnarray}
S_j & := & \left\{\beta \in \N_0^{|\Lambda|}: \beta_m = \alpha_m \:\: \mbox{for}\:\: m=1,\ldots,j-1, \right.  \\
& & \hspace{2cm}  \left. \beta_j \in \N_0 \setminus \{\alpha_j\}, \: \beta_m \in \N_0 \:\: \mbox{for} \:\: m=j+1,\ldots, |\Lambda| \right\}, \notag
\end{eqnarray}
$j= 1, \ldots, |\Lambda|$. For each $j$,
\begin{eqnarray} \label{kterm}
\sum_{\beta \in S_j} \prod_{\ell=1}^{|\Lambda|} W_{\alpha_\ell,\beta_\ell}^{\eta_{\ell,t}, \xi_\ell}
& = & \left( \prod_{\ell<j} W_{\alpha_\ell,\alpha_\ell}^{\eta_{\ell,t}, \xi_\ell} \right) \left( \sum_{\beta_j \not= \alpha_j} W_{\alpha_j, \beta_j}^{\eta_{j,t}, \xi_j} \right) \sum_{\beta_{j+1},\ldots, \beta_{|\Lambda|}} \left( \prod_{\ell>j} W_{\alpha_\ell, \beta_\ell}^{\eta_{\ell,t}, \xi_\ell} \right) \\
& = & \left( \prod_{\ell<j} W_{\alpha_\ell,\alpha_\ell}^{\eta_{\ell,t}, \xi_\ell} \right) \left( \sum_{\beta_j \not= \alpha_j} W_{\alpha_j, \beta_j}^{\eta_{j,t}, \xi_j} \right) \left(\prod_{\ell>j} \sum_{\beta_\ell \in \N_0} W_{\alpha_\ell, \beta_\ell}^{\eta_{\ell,t}, \xi_\ell} \right) \notag
\end{eqnarray}
We now use the fact that the Weyl operators are unitary and  the Weyl relations (\ref{eq:Weyl-Relations}) to get that for every $\ell\in\Lambda$,
\begin{equation} \label{simplebound}
|W_{\alpha_\ell,\alpha_\ell}^{\eta_{\ell,t},\xi_\ell}|\leq 1 \text{  and  }
\left|\sum_{\beta_\ell\in \N_0} W_{\alpha_\ell,\beta_\ell}^{\eta_{\ell,t},\xi_\ell}
\right|=
|\langle\alpha_\ell|\W_{\eta_{\ell,t}}\W_{\xi_{\ell}}|\alpha_\ell\rangle|
=
|\langle\alpha_\ell|\W_{\eta_{\ell,t}+\xi_\ell}|\alpha_\ell\rangle|
\leq 1.
\end{equation}
From \eqref{kterm} and \eqref{simplebound} we find that the absolute value of (\ref{corexpansion}) can be bounded by
\begin{equation}\label{Corre:proof:bound:sum}
|C^{\I}_\alpha(f,g,t)|\leq  \sum_{j=1}^{|\Lambda|} \ \
\sum_{\beta_j\in\mathbb{N}_0 \setminus \{\alpha_j\}}
|W_{\alpha_j,\beta_j}^{\eta_{j,t},\xi_j}|.
\end{equation}
In the following we deal with the sum over $\beta_j$ and, for simplicity, we suppress the subscripts $j$ and $t$.  Using the explicit formulas for the matrix elements of the Weyl operators in Lemma~\ref{lem:Weylmatrix}, we find
\begin{eqnarray}\label{Corre:proof:prod-bound}
|W_{\alpha,\beta}^{\eta,\xi}|&=&
\frac{\min\{\alpha ,\beta \}!}{\max\{\alpha ,\beta \}!}\left(\frac{|\eta \xi |}{2}\right)^{|\alpha -\beta|}
L^{(|\alpha -\beta |)}_{\min\{\alpha ,\beta \}}\left(\frac{|\eta |^2}{2}\right)
 L^{(|\alpha -\beta |)}_{\min\{\alpha ,\beta \}}\left(\frac{|\xi |^2}{2}\right)
e^{-\frac{1}{4}(|\eta |^2+|\xi|^2)} \\
&\leq& \binom{\max\{\alpha ,\beta \}}{\min\{\alpha ,\beta \}}\left(\frac{|\eta \xi |}{2}\right)^{|\alpha -\beta|} \notag
\end{eqnarray}
where we used that, for any $\alpha,n\in\mathbb{N}_0$ and $x\geq 0$, generalized Laguerre polynomials satisfy the bound (e.g.\ \cite{Handbook})
\begin{equation} \label{Laguerrebound}
|L^{(\alpha)}_n(x)|\leq \binom{\alpha+n}{n} e^{\frac{x}{2}}.
\end{equation}

To proceed with finding a bound for \eqref{needtobound}, we now take the suprema over $\alpha$ and $t$ in \eqref{Corre:proof:bound:sum} and estimate averages separately over
\begin{equation}
\Omega:=\{k \in (0,\infty)^{\Lambda}:\ \sup_t\sum_{j=1}^{|\Lambda|} |\eta_{j,t} \xi_j|\leq 1\}
 \end{equation}
 and its complement $\Omega^c$. Since $|C^{\I}_\alpha(f,g,t)|\leq 2$ as well as $|C^{\I}_\alpha(f,g,t)|\leq 2 |C^{\I}_\alpha(f,g,t)|^{1/(\kappa+1)}$, for all $f,g:\Lambda\rightarrow \mathbb{C}$ and $t\in\mathbb{R}$, we obtain
\begin{eqnarray} \label{Corre:proof:main-bound}
\mathbb{E}\left(\sup_{\alpha, t} |C^{\I}_\alpha(f,g,t)|\right) & \leq &
\mathbb{E}\left( \chi_{\Omega} \sup_{\alpha,t} |C^{\I}_\alpha(f,g,t)|\right)+2\ \mathbb{P}\left(\Omega^c \right) \\
& \le & 2 \mathbb{E}\left(\chi_{\Omega} \sup_{\alpha,t} |C^{\I}_\alpha(f,g,t)|^{1/(\kappa+1)} \right) + 2\mathbb{P}(\Omega^c)^{1/(\kappa+1)} . \notag
\end{eqnarray}
 For the second term in (\ref{Corre:proof:main-bound}),
Chebyshev's inequality gives
\begin{equation} \label{Corre:proof:term2}
\mathbb{P}\left(\Omega^c \right) \leq \mathbb{E}\left(\sup_t\sum_{j=1}^{|\Lambda|} |\eta_{j,t} \xi_j|\right).
\end{equation}

We now show that a similar bound holds for the first term in (\ref{Corre:proof:main-bound}). Here we will use the detailed bounds found in \eqref{Corre:proof:bound:sum}, \eqref{Corre:proof:prod-bound} and \eqref{Laguerrebound}.  On the set $\Omega$, using that $|\alpha-\beta|\geq 1$ for all the terms appearing in the sum (\ref{Corre:proof:bound:sum}), we get
\begin{equation}\label{Corre:proof:eta-xi-bound}
|\eta \xi|^{|\alpha-\beta|}\leq |\eta \xi|.
\end{equation}
Then in (\ref{Corre:proof:bound:sum}) we split the sum
and obtain, using (\ref{Corre:proof:prod-bound}), (\ref{Corre:proof:eta-xi-bound}) and $\|\alpha\|_\infty\le\kappa$,
\begin{eqnarray}
\sum_{\beta =0}^{\alpha-1}|W_{\alpha,\beta}^{\eta,\xi}|
&\leq& |\eta \xi|
 \sum_{\beta =0}^{\alpha-1}
 \binom{\alpha}{\beta}\left(\frac{1}{2}\right)^{\alpha-\beta}\leq \left(\frac{3}{2}\right)^\kappa\ |\eta\xi| \label{Coree:proof:beta:0-(alpha-1)}\\
\sum_{\beta =\alpha+1}^{\infty}|W_{\alpha,\beta}^{\eta,\xi}|
&\leq& |\eta \xi| \sum_{\beta=\alpha+1}^\infty \binom{\beta}{\alpha}\left(\frac{1}{2}\right)^{\beta-\alpha}\leq 2^{\kappa+1}\ |\eta \xi|
\label{Coree:proof:beta:(alpha+1)-infty}
\end{eqnarray}
where we used the identity
\begin{equation}
\sum_{n=\ell}^\infty \binom{n}{\ell}x^{n-\ell}=\frac{1}{(1-x)^{\ell+1}},\text{ for }-1<x<1.
\end{equation}
By substituting (\ref{Coree:proof:beta:0-(alpha-1)}) and (\ref{Coree:proof:beta:(alpha+1)-infty}) into (\ref{Corre:proof:bound:sum}) we obtain, uniformly in $\alpha$,
\begin{equation}\label{Corre:proof:term1}
\chi_{\Omega}|C^{\I}_\alpha(f,g,t)|^{1/(\kappa+1)} \leq 3 \left( \sum_{j=1}^{|\Lambda|} |\eta_{j,t} \xi_j| \right)^{1/(\kappa+1)}.
\end{equation}
Combining \eqref{Corre:proof:term1} and \eqref{Corre:proof:term2} into (\ref{Corre:proof:main-bound}), using Jensen's inequality $\mathbb{E}(X^s)\leq \mathbb{E}(X)^s$ for $0\leq s\leq 1$, yields
\begin{equation}
\mathbb{E}\left(\sup_t |C^{\I}_\alpha(f,g,t)|\right)\leq 8\ \left( \mathbb{E} \Big( \sup_t \sum_{j=1}^{|\Lambda|} |\eta_{j,t} \xi_j| \Big) \right)^{1/(\kappa+1)}.
\end{equation}
Thus we have reduced Theorem \ref{thm:Weyl-Eigen-ExpDecay} to

\begin{lem}\label{Corre:lem}
For any $f,g:\Lambda\rightarrow\mathbb{C}$ we have
\begin{equation} \label{Corre:lem:bound}
\mathbb{E}\left(\sup_{t\in\mathbb{R}}\sum_{j=1}^{|\Lambda|} |(V\X f_t)(j)(V\X g)(j)|\right)\leq C(1+\lambda_0^{1/2})^2 \sum_{x, y \in \Lambda}|f(x)g(y)| e^{-\mu|x-y|}.
\end{equation}
\end{lem}
\begin{proof} This is a variation of the proof of Lemma~\ref{lem:ave-dec-l2}. First, we observe that
\begin{equation} \label{jsum}
\sum_j |(V\X f_t)(j)V\X g)(j)|= \sum_j |\langle\idty_{\{j\}}V\X f_t, V\X g\rangle| = \sum_{j\in S_{\lambda_0}}|\langle \idty_{\{j\}} Vf, Vg \rangle|,
\end{equation}
having used $V\X f_t=\idty_{S_{\lambda_0}} e^{2it\gamma}Vf$ and $V\X g = \idty_{S_{\lambda_0}} Vg$, see (\ref{eq:Vinv-Chi-V}). Note that the supremum over $t$ has turned out to be trivial. We proceed by using that almost surely $h_\Lambda$ is non-degenerate, meaning that ${\mathcal O} \idty_{\{j\}} {\mathcal O}^T = \chi_{\{\gamma_j^2\}}(h_\Lambda)$ for all $j$. Thus, after expanding the right hand side of \eqref{jsum} using \eqref{def:Weyl:Vf}, we arrive at an upper bound for \eqref{jsum} which can almost surely be expressed as
\begin{eqnarray} \label{fourtermsum}
&\leq& \sum_{j\in S_{\lambda_0}} \left( |\langle\RE[f], \gamma_j^{-1} \chi_{\{\gamma_j^2\}}(h_\Lambda)\RE[g]\rangle|
+|\langle\RE[f],\chi_{\{\gamma_j^2\}}(h_\Lambda)\IM[g]\rangle|\right. \\
&&\left.\hspace{0.3cm}+|\langle\IM[f],\chi_{\{\gamma_j^2\}}(h_\Lambda)\RE[g]\rangle|
+|\langle\IM[f], \gamma_j \chi_{\{\gamma_j^2\}}(h_\Lambda)\IM[g]\rangle|\right). \notag
\end{eqnarray}
Consider the first term in this sum. Again by almost sure non-degeneracy of $h_\Lambda$, we can find a function $u: \{\gamma_j^2: j \in S_{\lambda_0}\} \to \C$ with $|u|\le 1$ such that
\begin{equation}
| \langle \RE[f], \gamma_j^{-1} \chi_{\{\gamma_j^2\}}(h_\Lambda) \RE[g] \rangle |= \langle \RE[f], \gamma_j u(\gamma_j^2) \RE[g] \rangle
\end{equation}
for all $j\in S_{\lambda_0}$, so that, setting $u=0$ elsewhere,
\begin{equation} \label{intstep}
\sum_{j\in S_{\lambda_0}} | \langle \RE[f], \gamma_j^{-1} \chi_{\{\gamma_j^2\}}(h_\Lambda) \RE[g] \rangle | = \langle \RE[f], h_\Lambda^{-1/2} u(h_\Lambda) \X \RE[g] \rangle.
\end{equation}
Taking expectations in \eqref{intstep} and using \eqref{def:eig-corr} we arrive at the bound
\begin{equation}
\sum_{x,y} |f(x)g(y)| \E \left( \sup_{|g|\le 1} |\langle \delta_x, h_\Lambda^{-1/2} u(h_\Lambda) \delta_y \rangle| \right) \le C \sum_{x,y} |f(x)g(y)| e^{-\mu|x-y|}.
\end{equation}
The other terms in \eqref{fourtermsum} can be treated similarly, where two terms pick up an extra factor $\lambda_0^{1/2}$ due to having to use \eqref{eig-corr-2} and one term picks up a $\lambda_0$ from \eqref{eig-corr-3}. Collecting all terms we arrive at \eqref{Corre:lem:bound}.
\end{proof}

\begin{appendix}

\section{Matrix entries of Weyl operators at the eigenstates} \label{sec:Weylops}

Here we provide expressions for the matrix elements $\langle\psi_\alpha,\W(f)\psi_\beta\rangle$ of the Weyl operators, explicitly describing them in terms of: (i) the ``Bogolubov tranformation''  \eqref{bfromqp}, encoded through the mapping $V$ in \eqref{def:Weyl:Vf}, which maps $H_\Lambda$ to an uncoupled oscillator system, and (ii) objects from the elementary theory of a single quantum oscillator and the corresponding one-dimensional Weyl operators (in particular Laguerre polynomivals).

We begin be reviewing what we need from the latter: The unique normalized ground state of the one-dimensional harmonic oscillator $p^2+q^2$ in $\mathcal{L}^2(\R)$ is
\begin{equation}
|0\rangle = \pi^{-1/4} e^{-\frac{1}{2} q^2}.
\end{equation}
This vaccum vector is the (up to a phase) unique solution of $a|0\rangle =0$, and all excited states are generated as
\begin{equation} \label{oscstates}
|k\rangle = \frac{1}{\sqrt{k!}} (a^*)^k |0\rangle, \quad k=1,2,\ldots,
\end{equation}
with the creation and annihilation operators $a^* = \frac{1}{\sqrt{2}} (q+ip)$ and $a= \frac{1}{\sqrt{2}} (q-ip)$.

The one-dimensional Weyl operators are, for $z\in \C$,
\begin{equation} \label{1DWeyl}
\W_z=\exp\left(\frac{i}{\sqrt{2}}(\bar{z} a+z a^*)\right).
\end{equation}
Among their properties are
\begin{equation} \label{fact1}
\langle 0|\W_z|0\rangle=e^{-\frac{1}{4}|z|^2},
\end{equation}
\begin{equation} \label{fact2}
[ a^\ell, \W_z ] = \sum_{j=1}^{\ell} { \ell \choose j} \left( \frac{i z}{ \sqrt{2}} \right)^j \W_z a^{\ell -j}, \quad \ell \ge 1,
\end{equation}
\begin{equation} \label{fact3}
\langle 0| \W_z (a^*)^\ell |0 \rangle = \left( \frac{i \cdot \overline{z}}{ \sqrt{2}} \right)^\ell \langle 0| \W_z|0\rangle, \quad \ell \ge 1.
\end{equation}
\eqref{fact1} can be seen from the Baker-Campbell-Hausdorff formula. \eqref{fact2} follows by induction after showing $[a,\W_z] = \frac{iz}{\sqrt{2}} \W_z$ via differentiation and integrating of $\W_{-tz} a \W_{tz}$ as a function of $t$. Applying this with $\W_{-z}$ and taking adjoints gives a related expression for $[\W_z, (a^*)^\ell]$, which yields \eqref{fact3}.

These rules allow to calculate the matrix elements of the Weyl operators in the oscillator eigenbasis:

\begin{lem} \label{lem:Weylmatrix}
Let $z\in\mathbb{C}$ and $0\le n \leq k$ be integers, then
\begin{equation}\label{eq:Local-Weyl:entries}
\langle n |\W_{z}|k\rangle=
\sqrt{\frac{n!}{k!}}
\left(\frac{i\overline{z}}{\sqrt{2}}\right)^{k-n}
L^{(k-n)}_{n}\left(\frac{|z|^2}{2}\right)e^{-\frac{1}{4}|z|^2}.
\end{equation}
Here, for $k, n\in\mathbb{N}_0$, $L^{(k)}_n(\cdot)$ is the $k$-generalized Laguerre polynomial
\begin{equation}\label{def:generalized-L}
L^{(k)}_n(z)=\sum_{j=0}^n \binom{n+k}{n-j}\frac{(-1)^j z^j}{j!},\text{ for }z\in\mathbb{C}.
\end{equation}
\end{lem}

For the case $0\leq k \leq n$, taking adjoints and using $\W_z^*=\W_{-z}$, this gives the related expression
\begin{equation} \label{eq:Local-Weyl:entries2}
\langle n |\W_{z}|k\rangle=
\sqrt{\frac{k!}{n!}}
\left(-\frac{i\overline{z}}{\sqrt{2}}\right)^{k-n}
L^{(n-k)}_{k}\left(\frac{|z|^2}{2}\right)e^{-\frac{1}{4}|z|^2}.
\end{equation}

\begin{proof}
Observe that (\ref{fact2}) implies
\begin{eqnarray} \label{simple_comm}
a^n  \W_z  (a^*)^k  & = & [a^n , \W_z ] (a^*)^k  + \W_z  a^n (a^*)^k   \\
& = & \sum_{j=1}^n  {n  \choose j} \left( \frac{i z}{ \sqrt{2}} \right)^j \W_z  a^{n -j}(a^*)^k  + \W_z  a^n  (a^*)^k \notag
\end{eqnarray}
and, as a result,
\begin{eqnarray}
\sqrt{n !} \sqrt{k !} \langle n| \W_z|  k \rangle & = &
 \langle 0| a^n  \W_z  (a^*)^k | 0 \rangle  \\
& = & \sum_{j=1}^n  {n  \choose j} \left( \frac{i z}{ \sqrt{2}} \right)^j \langle 0| \W_z  a^{n -j}(a^*)^k|  0 \rangle  +  \langle 0| \W_z a^n (a^*)^k|0 \rangle \nonumber \\
& = & \sum_{j=1}^n  {n  \choose j} \frac{k !}{(k -n +j)!} \left( \frac{i z}{ \sqrt{2}} \right)^j \langle 0| \W_z  (a^*)^{k -n +j} |0 \rangle+ \notag \\
&& \hspace{4cm} + \frac{k !}{(k -n )!}  \langle 0| \W_z  (a^*)^{k -n } |0 \rangle. \notag
\end{eqnarray}
An application of (\ref{fact3}) shows that
\begin{equation}
\langle n| \W_z | k \rangle = \frac{ \sqrt{k !}}{ \sqrt{n !}} \left(\frac{i \overline{z}}{\sqrt{2}}\right)^{k -n } \sum_{j=0}^n  { n  \choose j} \frac{(-1)^j}{(k -n +j)!} \left(\frac{|z|^2}{2}\right)^{j}\langle 0|\W_z|0\rangle,
\end{equation}
which gives the result by \eqref{eq:Local-Weyl:entries}, taking into account that
\begin{equation}
\binom{n}{j}\frac{1}{(k-n+j)!}=\frac{n!}{k!}\binom{n+(k-n)}{n-j}\frac{1}{j!}.
\end{equation}
\end{proof}

To express the matrix elements $\langle \psi_\alpha, \W(f) \psi_\beta \rangle$ in terms of products of matrix elements of one-dimensional Weyl operators of the explicit form \eqref{eq:Local-Weyl:entries}, we consider the product basis
\begin{equation} \label{prodbasis}
|\alpha \rangle = \otimes_{j=1}^{|\Lambda|} |\alpha_j\rangle = \left( \prod_{j=1}^{|\Lambda|} \frac{1}{\sqrt{\alpha_j!}}(a_j^*)^{\alpha_j} \right) |0^{\otimes |\Lambda|}\rangle, \quad \alpha \in \N_0^{|\Lambda|}
\end{equation}
 in $\bigotimes_{j=1}^{|\Lambda|} \mathcal{L}^2(\R)$. Here $|0^{\otimes |\Lambda|} \rangle = \otimes_{j=1}^{|\Lambda|} |0_j\rangle$ is the product vacuum and the (non-interacting) creation and annihilation operators
\begin{equation}
a_j=\frac{1}{\sqrt{2}}(q_j+ip_j), \text{ and } a_j^*= \frac{1}{\sqrt{2}}(q_j-ip_j), \quad j=1,\ldots,|\Lambda|
\end{equation}
satisfy the CCR.

\begin{lem}\label{thm:Weyl:entries}
For any $f:\Lambda\rightarrow\mathbb{C}$ and $\alpha,\beta\in\mathbb{N}_0^{|\Lambda|}$, the matrix entries of the Weyl operator $\W(f)$ at the eigenstates $\psi_\alpha$ and $\psi_\beta$ are given as
\begin{equation}\label{eq:Weyl:product}
\langle\psi_\alpha,\W(f)\psi_\beta\rangle=\prod_{j=1}^{|\Lambda|}\langle\alpha_j|\W_{(Vf)(j)}|\beta_j\rangle
\end{equation}
where $V$ is defined in (\ref{def:Weyl:Vf}).
\end{lem}

Note that, in turn, the factors on the right of \eqref{eq:Weyl:product} can be explicitly expressed through \eqref{eq:Local-Weyl:entries} and \eqref{eq:Local-Weyl:entries2}.

\begin{proof}
Let $U:\bigotimes_{j=1}^{|\Lambda|} \mathcal{L}^2(\R) \to \mathcal{H}_\Lambda$ be the unitary operator determined by $\psi_\alpha= U |\alpha\rangle$ for all $\alpha\in\mathbb{N}_0^{|\Lambda|}$.
By \eqref{eigenvectors} and \eqref{prodbasis} we have $b_j^* U|\alpha\rangle = Ua_j^*|\alpha\rangle$ and thus
\begin{equation}\label{eq:b-to-a}
U^* b_j U= a_j, \;\; U^* b_j^* U = a_j^* \;\text{ for all } j=1,\ldots,|\Lambda|.
\end{equation}
Therefore \eqref{defWeyl} yields
\begin{equation}
U^*\W(f)U=\exp\left(\frac{i}{\sqrt{2}}\sum_{j=1}^{|\Lambda|} (\overline{(Vf)(j)}a_j+(Vf)(j)a_j^*)\right)=\bigotimes_{j=1}^{|\Lambda|} \W_{(Vf)(j)},
\end{equation}
where we have factored into one-dimensional Weyl operators of the form \eqref{1DWeyl}. Finally,
\begin{equation} \label{eq:Weyl:matrix-elements-product}
\langle\psi_\alpha,\W(f)\psi_\beta\rangle =
\langle \alpha|U^*\W(f)U|\beta \rangle = \langle \alpha|\bigotimes_{j=1}^{|\Lambda|} \W_{(Vf)(j)}|\beta \rangle
=\prod_{j=1}^{|\Lambda|} \langle\alpha_j|\W_{(Vf)(j)}|\beta_j\rangle.
\end{equation}
\end{proof}

\section{Non-degeneracy of the spectrum} \label{sec:nondeg}

Here we prove \eqref{simplespec}, i.e., that the disordered oscillator systems considered above almost surely have simple spectrum, i.e., all eigenvalues are non-degenerate.

For a finite box $\Lambda$ let
\begin{equation}
H(k) = \sum_{x\in\Lambda}p_x^2 + q^T (h_{0} + k) q,
\end{equation}
where $h_{0} = h_{0,\Lambda}$ is the discrete Laplacian \eqref{graphLap} on $\Lambda$ (as $\Lambda$ is fixed in this section we will drop the subscript) and $k\in (0,\infty)^{\Lambda}$ is understood as a multiplication operator. Due to the positivity of the numbers $k_x$, it is well known that $H(k)$ is strictly positive definite with purely discrete spectrum, e.g.\ \cite{RS2}.

\begin{lem} \label{lem:simple}
The spectrum of $H(k)$ is simple for Lebesgue almost every $k\in (0,\infty)^{\Lambda}$.
\end{lem}

In our applications $k=(k_x)_{x\in\Lambda}$ are non-negative i.i.d.\ random variables with absolutely continuous distribution. Thus Lemma~\ref{lem:simple} implies immediately the almost sure non-degeneracy of $H_\Lambda$ claimed in \eqref{simplespec}.

To prove Lemma~\ref{lem:simple}, we proceed similar to the proof of a corresponding fact for quantum spin systems in Proposition~A.1 of \cite{ARS15}, using two main steps:

(i) There exists at least one $k \in \R^{\Lambda}$ such that the spectrum of $H(k)$ is simple. 

(ii) Lemma~\ref{lem:simple} follows from (i) by an analyticity argument.

The main difference to the argument in \cite{ARS15} is in step (i), mostly because we are considering unbounded operators here. To prove (i), we will use \eqref{eigenvalues}, i.e., that the eigenvalues of $H(k)$ are given by $\sum_j \gamma_j (2\alpha_j+1)$, $\alpha \in \N_0^{|\Lambda|}$, where $\gamma_j^2$, $j=1,\ldots,|\Lambda|$ are the eigenvalues of $h_{0}+k$. It therefore suffices to prove the existence of $k \in (0,\infty)^{\Lambda}$ such that the square roots of the eigenvalues of $h_{0}+k$ are rationally independent. In fact, we will prove the following:

Let $E_1(k) \le E_2(k) \le \ldots \le E_{|\Lambda|}(k)$ be the ordered eigenvalues of $h_{0}+k$, counted with multiplicity, and $\vec{E}:\R^{\Lambda} \to \R^{|\Lambda|}$ given by $\vec{E}(k) = (E_1(k),\ldots, E_{|\Lambda|}(k))$. Then the range of $\vec{E}$ contains a non-trivial open subset.

This implies (i) because the set of all $a=(a_j)_{1\le j \le |\Lambda|} \in (0,\infty)^{|\Lambda|}$ such that the vector $(\sqrt{a_j})_{1\le j \le |\Lambda|}$ is rationally independent is dense in $(0,\infty)^{|\Lambda|}$.

Choose $k^{(0)} \in (0,\infty)^{|\Lambda|}$ with components such that
\begin{equation}
0 < k_1^{(0)} < k_2^{(0)} < \ldots < k_{|\Lambda|}^{(0)}
\end{equation}
and $|k_{j+1}^{(0)} - k_j^{(0)}| > 2 \|h_0\| + C$ for all $j$ and a constant
\begin{equation} \label{Cbound}
C> \frac{144 \|h_0\| |\Lambda|^2}{\pi}
\end{equation}
(for reasons which will become clear at the end of the following calculation). Thus, e.g.\ by the variational principle for eigenvalues, the eigenvalues of $h_0+k^{(0)}$ are simple with $|E_j(k^{(0)}) - k_j^{(0)}| \le \|h_0\|$ and $|E_{j+1}(k^{(0)}) - E_j(k^{(0)})| > C$ for all $j$. They are analytic functions in $k$ for $k$ near $k^{(0)}$ and we will show that the Jacobian $\frac{\partial \vec{E}}{\partial k}$ is non-singular at $k=k^{(0)}$. Thus $\vec{E}$ is locally invertible near $k^{(0)}$ by the inverse function theorem, and thus its range contains an open neighborhood of $\vec{E}(k^{(0)})$.

By first order perturbation theory (Feynman-Hellmann) the Jacobian at $k^{(0)}$ has matrix elements
\begin{equation} \label{FeynHell}
\frac{\partial E_j}{\partial k_\ell}(k^{(0)}) = \langle v_j, I_{\ell} v_j \rangle = |v_j(\ell)|^2,
\end{equation}
where $v_j$ is a normalized eigenvector of $h_0+k^{(0)}$ to $E_j(k^{(0)})$ and $I_{\ell}$ the matrix with a single one in the $(\ell,\ell)$-th entry. The eigenvectors $v_j$ will be expressed through the orthogonal projections $P_j$ onto the eigenspace to $E_j(k^{(0)})$. The counterclockwise rectangular contour $\Gamma$ with vertices
\begin{equation}
k_j^{(0)}+\|h_0\|+\frac{C}{2} +i\frac{C}{2}, k_j^{(0)}-\|h_0\|-\frac{C}{2} +i\frac{C}{2}, k_j^{(0)}-\|h_0\|-\frac{C}{2} -i\frac{C}{2}, k_j^{(0)}+\|h_0\|+\frac{C}{2}-i\frac{C}{2}
\end{equation}
contains no other eigenvalues of $h_0+k^{(0)}$, so that
\begin{eqnarray} \label{Pjcont}
P_j & = & \frac{1}{2\pi i} \int_{\Gamma} (h_0+k^{(0)}-z)^{-1}\,dz \\
& = & \frac{1}{2\pi i} \int_{\Gamma} ((k^{(0)}-z)^{-1} - (k^{(0)}-z)^{-1} h_0 (h_0+k^{(0)}-z)^{-1}) \,dz \notag \\
& = & I_j - \frac{1}{2\pi i} \int_{\Gamma} (k^{(0)}-z)^{-1} h_0 (h_0+k^{(0)}-z)^{-1} \,dz. \notag
\end{eqnarray}
The norm of the integrand is bounded by $4\|h_0\| /C^2$ (using that $\Gamma$ has distance at least $C/2$ from all eigenvalues of $k^{(0)}$ and $h_0+k^{(0)}$), and the length of $\Gamma$ is bounded by $2C + 4\|h_0\| \le 6C$. Thus \eqref{Pjcont} implies that
\begin{equation}
\|P_j - I_j\| \le \frac{1}{2\pi} \cdot 6C \cdot \frac{4\|h_0\|}{C^2} \le \frac{12 \|h_0\|}{C\pi}.
\end{equation}
From this we conclude that the normalized eigenvectors of $h_0+k^{(0)}$ can be expressed as
\begin{equation}
v_j = \frac{P_j e_j}{\|P_j e_j\|} = \frac{e_j + r_j}{\|e_j+ r_j\|}
\end{equation}
with $\|r_j\| \le 12 \|h_0\|/(C\pi)$. Now a simple calculation shows
\begin{equation}
\|v_j-e_j\| \le \frac{48 \|h_0\|}{C\pi}.
\end{equation}
Inserting into \eqref{FeynHell} yields, recall \eqref{Cbound},
\begin{equation}
\left| \frac{\partial E_j}{\partial k_{\ell}} (k^{(0)}) - \delta_{\ell j} \right| \le \frac{96 \|h_0\|}{C\pi} + \left( \frac{48 \|h_0\|}{C\pi} \right)^2 \le \frac{144 \|h_0\|}{C\pi},
\end{equation}
so that
\begin{equation}
\left\| \frac{\partial \vec{E}}{\partial k}(k^{(0)}) - \idty \right\| \le \frac{144 \|h_0\| |\Lambda|^2}{C\pi} < 1.
\end{equation}
Thus the Jacobian at $k^{(0)}$ is invertible, which completes the proof of property (i).

To mimic the iterative analyticity argument from \cite{ARS15} for the proof of (ii), we observe that for each $x\in \Lambda$ and fixed numbers $k_y>0$, $y\not= x$, the eigenvalues of $H(k)$ can be labeled as functions $E_\alpha(k_x)$, $\alpha \in \N_0^\Lambda$, which are analytic in $k_x>0$. To see this note that by finite-dimensional analytic perturbation theory, for fixed $k_y$, $y\not= x$, the eigenvalues of $h_0+k>0$ can be labeled as analytic functions $\gamma_j^2(k_x)$, $j=1,\ldots, |\Lambda|$. Thus the same holds for their positive square roots $\gamma_j(k_x)$. Using \eqref{eigenvalues}, this gives the asserted labeling of the eigenvalues of $H(k)$ as $E_\alpha(k_x) = \sum_{j=1}^{|\Lambda|} \gamma_j(k_x)(2\alpha_j+1)$.

With this we can complete the proof of Lemma~\ref{lem:simple} with exactly the same iterative analyticity argument as in Step~2 of the proof of \cite[Lemma A2]{ARS15}: Let $k^{(0)} = (k_1^{(0)},\ldots,k_{|\Lambda|}^{0)})$ be as found in step (i). Then there is a nullset $N_1 \subset (0,\infty)$ such that  the eigenvalues of $H(k_1, k_2^{(0)},\ldots,k_{|\Lambda|}^{0)})$ are pairwise distinct for all $k_1 \not\in N_1$ (their eigenvalues $E_\alpha(k_1)$, $\alpha \in \N_0^\Lambda$, as functions of $k_1$ are pairwise distinct at $k_1=k_1^{(0)}$, so each of the countably many pairs $(E_\alpha (k_1), E_\beta (k_1))$, $\alpha \not= \beta$, can coincide for at most countably many values of $k_1$).

Fix any $k_1 \in (0,\infty)\setminus N_1$. We can now argue as above to get the existence of a set $N_2(k_1)$ such that for every $k_2 \in (0,\infty) \setminus N_2(k_1)$ all eigenvalues of $H(k_1,k_2, k_3^{(0)}, \ldots, k^{(0)}_{|\Lambda|})$ are distinct. By Fubini this means that $H(k_1,k_2, k_3^{(0)}, \ldots, k^{(0)}_{|\Lambda|})$ is simple for Lebesgue-a.e.\ $(k_1,k_2) \in (0,\infty)^2$. From here one proceeds iteratively to complete the proof of (ii).

\section{Energy density in the regime of localized excitations} \label{App:energydensity}

The goal of this appendix is to make more precise that the regime of localized excitations, i.e., the range of the spectral projection $P_{\I}$ of $H_\Lambda$ for which we have shown MBL properties in this work, contains states of positive energy density.

For simplicity, consider  cubes $\Lambda = \Lambda_L = [-L,L]^{\nu}$, fix a positive integer $\kappa$ and let $\alpha^{(\kappa)}$ be the occupation number vector in $\I$ with $\kappa$ excitations in all the sites of $S_{\lambda_0}$, i.e.,
\begin{equation} \label{maxenergy}
\alpha_j^{(\kappa)} = \left\{ \begin{array}{ll} \kappa, & \mbox{if $j\in S_{\lambda_0}$}, \\ 0, & \mbox{else}. \end{array} \right.
\end{equation}
The average energy density (per system size $|\Lambda|$) of $\psi_{\alpha^{(\kappa)}}$ can be explicitly characterized in terms of the density of states (DOS) $n(\lambda)$ of the infinite volume Anderson model $h=h_0+k$ on $\ell^2(\Z^{\nu})$:

\begin{proposition} \label{prop:energydensity}
Under the assumption \eqref{density} it holds that
\begin{equation} \label{energydensity}
\lim_{L\to \infty} \frac{1}{|\Lambda_L|} \E (E_{\alpha^{(\kappa)}} - E_0) = 2\kappa \int_0^{\lambda_0} n(\lambda) \lambda^{1/2}\,d\lambda.
\end{equation}
\end{proposition}

Recall that $n(\lambda)=N'(\lambda)$, where the integrated density of states (IDS) is
\begin{equation}
N(\lambda) = \E (\langle \delta_0, \chi_{(-\infty,\lambda]}(h) \delta_0 \rangle).
\end{equation}
From \eqref{density} it follows that $N(\lambda)$ is absolutely continuous, e.g.\ \cite[Corollary~5.24]{Kirsch}, so that $n(\lambda)$ exists almost everywhere. In fact, if the density $\rho$ in \eqref{density} satisfies ess-$\inf_{k\in [0,k_{max}]} \rho(k) >0$,  then $n(\lambda)$ is strictly positive almost everywhere on the almost sure spectrum $[0, 4\nu+k_{max}]$ of $h$ and uniformly bounded away from zero on $(\delta, 4\nu+k_{max}-\delta)$ for every $\delta>0$, see \cite{HislopMueller}.

The therefore strictly positive quantity $2\kappa \int_0^{\lambda_0} n(\lambda) \lambda^{1/2}\,d\lambda$ takes the role of the averaged maximal energy density of the states where we have proven MBL properties in Theorems~\ref{thm:quasiloc} and \ref{thm:Weyl-Eigen-ExpDecay}. The fact that $\kappa$ can be any integer shows that there are many-body localized states with arbitrarily high energy density, but the bounds also show the price one has to pay for large $\kappa$.

It is not hard to guess \eqref{energydensity} from \eqref{eigenvalues}. The main issue we have to deal within its {\it proof} is that we need to relate the infinite volume IDS and DOS to the finite volume eigenvalue counting function. In finite volume we can use Dirichlet-Neumann bracketing, but some care is needed to get two-sided bounds.

Note that choosing $h_\Lambda =h_{0,\Lambda} + k$ via the graph Laplacian \eqref{graphLap} is the discrete analogue of using Neumann boundary conditions. In particular, this means that $h_\Lambda \le h_{\Lambda_1} \oplus h_{\Lambda_2}$ for any disjoint decomposition $\Lambda = \Lambda_1 \cup \Lambda_2$. This has the consequence that the infinite volume IDS is related to the finite volume counting function $N_\Lambda(\lambda) := \tr \chi_{(-\infty,\lambda)}(h_\Lambda) = |\{x: \gamma_x^2 < \lambda\}|$ by (e.g.\ \cite{Kirsch})
\begin{equation} \label{Nbound}
N_\Lambda(\lambda) \ge |\Lambda| N(\lambda).
\end{equation}
The correct choice of Dirichlet boundary conditions, from a quadratic form point of view, is to set $h_{0,\Lambda}^{D} := h_{0,\Lambda} + 2(2\nu - n_\Lambda)$, where $n_\Lambda(x):= |\{y\in \Lambda: \|x-y\|_1=1\}|$ is the degree function on the subgraph $\Lambda$ of $\Z^{\nu}$. For $h_\Lambda^{D} := h_{0,\Lambda}^{D} + k$ this leads to $h_{\Lambda}^{D} \ge h_{\Lambda_1}^{D} \oplus h_{\Lambda_2}^{D}$ for disjoint decompositions and consequently
\begin{equation} \label{Dbound}
N_\Lambda^{D}(\lambda) \le |\Lambda| N(\lambda)
\end{equation}
for the Dirichlet eigenvalue counting function $N_\Lambda^D(\lambda) = \tr \chi_{(-\infty,\lambda)}(h_{\Lambda}^{D})$, see \cite{Kirsch}.

Note that $2\nu-n_\Lambda$ is non-zero only on the boundary of $\Lambda$, so that the operators $h_\Lambda$ and $h_\Lambda^D$ differ by an operator of rank bounded by $CL^{\nu-1}$. This means that
\begin{equation} \label{DNbound}
N_\Lambda(\lambda) \le N_\Lambda^{D}(\lambda) + CL^{\nu-1}
\end{equation}
uniformly in $\lambda$.

We now proceed with the proof of \eqref{energydensity}.
Let $n$ be a positive integer, to be specified later. For $\alpha$ as in \eqref{maxenergy} we have by \eqref{eigenvalues} that
\begin{eqnarray} \label{firstone}
E_\alpha - E_0 & = & 2\kappa \sum_{x\,:\, \gamma_x^2 < \lambda_0} \gamma_x =  2\kappa \sum_{j=1}^n \sum_{x\,: \,\gamma_x^2 \in [(j-1)\lambda/n, j\lambda_0/n)} \gamma_x  \\
& \ge & 2\kappa \sum_{j=1}^n \left(N_\Lambda \left(\frac{j}{n}\lambda_0\right) - N_\Lambda\left(\frac{j-1}{n} \lambda_0\right)\right) \left(\frac{j-1}{n} \lambda_0\right)^{1/2}. \notag
\end{eqnarray}
For each $j$ we have by \eqref{DNbound}, \eqref{Nbound} and \eqref{Dbound} that
\begin{eqnarray} \label{Ndiffbound}
N_\Lambda \left(\frac{j}{n}\lambda_0\right) - N_\Lambda\left(\frac{j-1}{n} \lambda_0\right) & \ge & N_\Lambda\left( \frac{j}{n} \lambda_0\right) - N_\Lambda^{D} \left( \frac{j-1}{n} \lambda_0 \right) - C L^{\nu-1}  \\
& \ge & |\Lambda| \left(N\left(\frac{j}{n} \lambda_0 \right) - N\left(\frac{j-1}{n} \lambda_0 \right) \right) - CL^{\nu-1}. \notag
\end{eqnarray}
Note that the final expression in \eqref{Ndiffbound} is non random. Taking expectations in \eqref{firstone} and bounding $\sum_{j=1}^n ((j-1)\lambda_0/n)^{1/2} \le n \lambda_0^{1/2}$ gives
\begin{equation} \label{secondone}
\E(E_\alpha - E_0) \ge 2\kappa |\Lambda| \sum_{j=1}^n \left( N \left(\frac{j}{n}\lambda_0 \right) - N \left( \frac{j-1}{n} \lambda_0 \right) \right) \left( \frac{j-1}{n} \lambda_0 \right)^{1/2} - 2\kappa CL^{\nu-1} n \lambda_0^{1/2}.
\end{equation}
The first term on the right is a Riemann sum for the integral in \eqref{energydensity}. However, in order to not be doomed by the second term in \eqref{secondone}, we have to couple the $L\to \infty$ and $n\to\infty$ limits by now choosing $n$ to be the integer closest to $\sqrt{L}$. Then the second term in \eqref{secondone} is of order $L^{\nu- 1/2}$. The first term is equal to $2\kappa |\Lambda| \int_0^{\lambda_0} n(\lambda) \lambda^{1/2}\,d\lambda$ up to
\begin{equation} \label{thirdone}
2\kappa |\Lambda| \sum_{j=1}^n \int_{(j-1)\lambda_0/n}^{j\lambda_0/n} n(\lambda) \left( \left( \frac{j-1}{n} \lambda_0 \right)^{1/2} - \lambda^{1/2} \right)\,d\lambda.
\end{equation}
Using that $|((j-1)\lambda_0/n)^{1/2} - \lambda^{1/2}| \le (\lambda_0/n)^{1/2} \le C \lambda_0^{1/2} L^{-1/4}$ uniformly in $j$ and $\lambda \in [(j-1)\lambda_0/n, j\lambda_0/n]$, we see that \eqref{thirdone} is bounded in absolute value by $C|\kappa| \lambda_0^{1/2} N(\lambda_0) L^{\nu - 1/4}$, i.e., a term of order $L^{\nu-1/4}$.

As all perturbations which have appeared are of order lower than $|\Lambda| \sim L^\nu$, we therefore find
\begin{equation} \label{fourthone}
\liminf_{L\to\infty} \frac{\E(E_\alpha-E_0)}{|\Lambda|} \ge 2\kappa \int_0^{\lambda_0} n(\lambda) \lambda^{1/2}\,d\lambda.
\end{equation}

To get a corresponding upper bound on the $\limsup$, we start from
\begin{equation}
E_\alpha - E_0 \le  2\kappa \sum_{j=1}^n \left(N_\Lambda \left(\frac{j}{n}\lambda_0\right) - N_\Lambda\left(\frac{j-1}{n} \lambda_0\right)\right) \left(\frac{j}{n} \lambda_0\right)^{1/2}
\end{equation}
and then use
\begin{eqnarray}
N_\Lambda \left(\frac{j}{n}\lambda_0\right) - N_\Lambda\left(\frac{j-1}{n} \lambda_0\right) & \le & N_\Lambda^{D}\left( \frac{j}{n} \lambda_0\right) - N_\Lambda \left( \frac{j-1}{n} \lambda_0 \right) + CL^{\nu-1} \\
& \le & |\Lambda| \left(N\left(\frac{j}{n} \lambda_0 \right) - N\left(\frac{j-1}{n} \lambda_0 \right) \right) + CL^{\nu-1}. \notag
\end{eqnarray}
With this one proceeds essentially as above and arrives at
\begin{equation}
\limsup_{L\to\infty} \frac{\E(E_\alpha-E_0)}{|\Lambda|} \le 2\kappa \int_0^{\lambda_0} n(\lambda) \lambda^{1/2}\,d\lambda,
\end{equation}
completing the proof of \eqref{energydensity}.

\end{appendix}


\end{document}